\documentclass{llncs}
\usepackage{fullpage}
\usepackage{amsmath,amsfonts,xcolor,stmaryrd}
\usepackage{tikz}
\usetikzlibrary{shapes}
\usetikzlibrary{positioning,automata}
\usetikzlibrary{arrows, automata,calc,shapes.arrows, backgrounds, fadings,intersections,  decorations.markings, positioning, snakes}
\usetikzlibrary{shadows}
\usetikzlibrary{calc,decorations.pathreplacing}
\tikzstyle{every state}=[scale=0.9, inner sep=1pt,minimum size=15pt]
\usepackage[utf8]{inputenc}
\usepackage{todonotes}
\usepackage{wrapfig}
\usepackage{subfig}
\usepackage{ael}
\usepackage[colorlinks=true,
citecolor=blue,linkcolor=red,urlcolor=black]{hyperref}
\usepackage{booktabs}
\usepackage{thm-restate}

\begin{document}

\title{Bounding Average-energy Games} 
\author{%
  Patricia Bouyer\inst{1}\thanks{Supported by ERC project EQualIS (308087).}
 \and 
  Piotr Hofman\inst{1,2} 
 \and 
  Nicolas Markey\inst{1,3}$^\star$
 \and
  Mickael Randour\inst{4}\thanks{F.R.S.-FNRS postdoctoral researcher.}
 \and 
  Martin Zimmermann\inst{5}\thanks{Supported by the DFG project TriCS (ZI 1516/1-1).}} 
\institute{%
  LSV, CNRS \& ENS Cachan, Universit\'e Paris Saclay, France
 \and 
University of Warsaw, ul.\ Banacha 2, 02-097 Warszawa, Poland
 \and
  IRISA, CNRS \& INRIA \& U. Rennes 1, France
 \and
  Computer Science Department, ULB - Universit\'e libre de Bruxelles, Belgium
 \and 
  Reactive Systems Group, Saarland University, 66123 Saarbrücken, Germany}
\maketitle

\begin{abstract}
  We consider average-energy games, where the goal is to minimize the
  long-run average of the accumulated energy. While several results
  have been obtained on these games recently, decidability of
  average-energy games with a lower-bound constraint on the energy
  level (but no upper bound) remained open; in~particular, so~far
  there was no known upper bound on the memory that is required for
  winning strategies.

  By reducing average-energy games with lower-bounded energy to
  infinite-state mean-payoff games and analyzing the density of
  low-energy configurations, we show an almost tight doubly-exponential upper
  bound on the necessary memory, and that the winner of average-energy
  games with lower-bounded energy can be determined in
  doubly-exponential time.  We~also prove \EXPSPACE-hardness of this
  problem.

  Finally, we consider multi-dimensional extensions of all types of
  average-energy games: without bounds, with only a lower bound, and
  with both a lower and an upper bound on the energy. We~show that the
  fully-bounded version is the only case to remain decidable in
  multiple dimensions.
\end{abstract}

\section{Introduction}
\label{sec-intro}
Quantitative two-player games of infinite duration provide a natural
framework for synthesizing controllers for reactive systems with resource restrictions in
an antagonistic environment (see e.g.,~\cite{BCHJ09,Ran13}). In such games, player~$\pO$ (who
represents the system to be synthesized) and player $\pI$ (representing
the antagonistic environment) construct an infinite path by moving a
pebble through a graph, which describes the interaction between the
system and its environment. The objective, a subset of the infinite
paths that encodes the controller's specification,
determines the winner of such a play. Quantitative games extend this
classical model by weights on edges for modeling costs, consumption,
or rewards, and by a quantitative objective to encode the
specification in terms of the weights. 

\begin{wrapfigure}{r}{.3\textwidth}
\centering
\begin{tikzpicture}[auto,node distance=2 cm, 
 thick]
\node[carre,jaune] 	(a)	at (0,1.8)				 	{$s_2$};
\node[rond,jaune]	(b)	at (1.5,0)	{$s_1$};
\node[rond,jaune]	(c)	at (-1.5,0)	{$s_0$};

\path[use as bounding box] (a);

\path[-latex']
	(-2.3,0) edge [] node [] {} (c)
	(a) edge [bend left] 	node[below left=-2pt] 	{$\scriptstyle  -4$}	(b)
	(a) edge [bend right] 	node[below right=-2pt]	{$\scriptstyle 2$}	(c)
	(b) edge [bend right = 15] 	node[above] 	{$\scriptstyle  0$}	(c)
	(c) edge [bend right = 15] 	node[below] 	{$\scriptstyle  -2$}	(b)
	(b) edge [bend right = 85] 	node[above right=-2pt]	{$\scriptstyle  -1$}	(a)
	(c) edge [bend left = 85] 	node[above left=-2pt]	{$\scriptstyle  4$}	(a)
; 
	\end{tikzpicture} 
      \caption{Simple weighted game.}
\label{fig:sampleGame}
	\vspace{-15pt}
\end{wrapfigure}
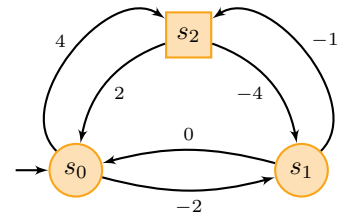

As an example, consider the game in Fig.~\ref{fig:sampleGame}: we
interpret negative weights as energy consumption and correspondingly
positive weights as recharges. Then, $\pO$ (who moves the pebble at
the circular states) can always maintain an energy level (the~sum of
the weights seen along a play prefix starting with energy zero) between
zero and six using the following strategy: when at state~$s_0$ with
energy level at least two, go~to state~$s_1$, otherwise go to
state~$s_2$ in order to satisfy the lower bound. At~state~$s_1$, 
always move to~$s_0$. It~is straightforward to verify that the
strategy has the desired property when starting at the initial
state~$s_0$ with initial energy zero. Note that this strategy requires
memory to be implemented, as its choices depend on the current energy
level and not only on the state the pebble is currently~at.

Formally, the \textit{energy objective} requires $\pO$ to maintain an energy level within some (lower and/or upper) bounds, which are either given as input or existentially quantified. In the example above, $\pO$ has a strategy to win for the  energy objective with lower bound zero and upper bound six. Energy objectives~\cite{bouyer2008,Chatterjee2013,Juhl13,DBLP:conf/icalp/JurdzinskiLS15} and their combinations with parity objectives~\cite{DBLP:journals/tcs/ChatterjeeD12,Chatterjee2013} have received significant attention in the literature.

However, a plain energy (parity) objective is sometimes not sufficient to
adequately model real-life systems. For example, consider the following
specification for the controller of an oil pump, based on a case study~\cite{hscc2009-CJLRR}: it has to keep the amount of oil in an accumulator within
given bounds (an energy objective with given lower and upper bounds)
while keeping the average amount of oil in the
accumulator below a given threshold \emph{in the long run}. The latter requirement reduces the wear and tear of the
system, but cannot be expressed as an energy objective nor as a parity
objective. Constraints on the long-run average energy level (which exactly represents the amount of oil in our example) can be specified using the \textit{average-energy
objective}~\cite{BMRLL16}. As seen in this example, they are typically studied in conjunction with bounds
on the energy level.

Recall the example in Fig.~\ref{fig:sampleGame}. The aforementioned strategy for $\pO$ guarantees a long-run average energy level, i.e., \textit{average-energy}, of at most~$11/4$ (recall we want to minimize it): the outcome with worst average-energy is~$\big(s_0 s_2 (s_0 s_1)^3\big)^\omega$, with energy levels $(4,6,4,4,2,2,0,0)^\omega$.

The average-energy objective was first introduced by Thuijsman and Vrieze
under the name \emph{total-reward}~\cite{TV87} (there is an unrelated,
more standard, objective called total-reward, see~\cite{BMRLL16} for a
discussion). Recently, the average-energy objective was independently
revisited by Boros \textit{et~al.}~\cite{BEGM15} and by Bouyer
\textit{et~al.}~\cite{BMRLL16}. The~former work studies Markov decision
processes and stochastic games with average-energy objectives.
The latter studies non-stochastic games with
average-energy objectives, with or without lower and upper
bounds on the energy level; it~determines the complexity of computing
the winner and the memory requirements for winning
strategies in such games.
In~particular, it~solves games with average-energy
objectives with \textit{both} an upper and a lower bound on the energy level by a
reduction to mean-payoff games: to~this~end, the~graph is extended to
track the energy level between these bounds (a~losing sink for
$\pO$ is reached if these bounds are exceeded). Thus, the~bounds
on the energy level are already taken care of and the average-energy
objective can now be expressed as a \textit{mean-payoff
objective}~\cite{EM79}, as the new graph encodes the
current energy level in its weights. This reduction yields an exponential-time decision algorithm. Moreover, it~is shown in~\cite{BMRLL16} that these games are
indeed \EXPTIME-complete. Note that the algorithm crucially depends on
the upper bound being given as part of the input, which implies that
the graph of the extended game is still finite.

One problem left open in~\cite{BMRLL16} concerns average-energy
games \textit{with only a lower bound} on the energy level:
computing the winner is shown to be \EXPTIME{}-hard, but it is 
problem is decidable at~all. Similarly, pseudo-polynomial lower bounds (i.e.,
lower bounds that are polynomial in the \emph{values} of the weights,
but possibly exponential in the size of their binary representations) on
the necessary memory to implement a winning strategy for $\pO$
are given, but no upper bound is known. The major obstacle toward
solving these problems is that without an upper
bound on the energy, a~strategy might allow arbitrarily large
energy levels while still maintaining a bounded average, by enforcing long stretches with a small energy level to offset the
large levels.

A step toward resolving these problems was taken by considering two variants of energy and average-energy objectives where (i) the upper bound on the energy level, or 
(ii) the threshold on the average-energy, is existentially
quantified~\cite{LLZ15}. It~turns out that these two variants are
equivalent. One direction is trivial: if~the~energy is bounded, then the average-energy is bounded. On the other hand, if $\pO$ can guarantee some upper bound on the
average, then he can also guarantee an upper bound on the energy
level, i.e., an (existential) average-energy objective can always be satisfied with bounded energy levels. This is shown by transforming a strategy satisfying a bound on the average (but possibly allowing arbitrarily high energy levels) into one that bounds the energy by skipping parts of plays where the energy level is much higher than the threshold on the average. However, the proof is not effective:
it~does not yield an upper bound on the necessary energy
level, just a guarantee that some bound exists. Even more so, it is still possible that the average has to increase when keeping the energy bounded. Hence, it does not answer our open problem: does achieving a \textit{given} threshold on the average-energy require unbounded energy levels and infinite memory?

Another potential approach toward solving the problem is to extend
the reduction presented in~\cite{BMRLL16} (which goes from average-energy games with both lower and upper bound on the energy level to
mean-payoff games) to games without such an upper bound, which results
in an infinite graph. This graph can be seen as the configuration
graph of a one-counter pushdown system, i.e., the stack height
corresponds to the current energy level, and the average-energy
objective is again transformed into a mean-payoff objective, where the
weight of an edge is given by the stack height at the target of the
edge. Hence, the weight function is unbounded. To the best of our
knowledge, such mean-payoff games have not been studied before.
However, mean-payoff games on pushdown systems with bounded weight 
functions are known to be undecidable~\cite{CV12}.

\paragraph{Our Contribution.}
This paper gives a full presentation of the results published in its conference version~\cite{fossacs2017}. We develop the first algorithm for
solving games with average-energy objectives and  a lower bound
on the energy level, and give
an upper bound on the necessary
memory to implement a winning strategy for~$\pO$ in such games.

First, we present an algorithm solving such games in doubly-exponential time (for the case of a binary encoding of the weights). The algorithm is based on the characterization of an average-energy game as a mean-payoff game on an infinite graph described above. If the average-energy of a play is bounded by the threshold~$t$, then configurations with energy level at most~$t$ have to be visited frequently. As there are only finitely many such configurations, we obtain cycles on this play. By a more fine-grained analysis, we obtain such a cycle with an average of at most~$t$ and whose length is bounded exponentially. Finally, by~analyzing strategies for reachability objectives in pushdown games, we show that $\pO$ can ensure that the distance between such cycles is bounded doubly-expo\-nen\-tially. From these properties, we obtain a doubly-expo\-nen\-tial upper bound on the necessary energy level to ensure an average-energy of at most~$t$. The~resulting equivalent average-energy game with a lower and an upper bound can be solved in doubly-exponential time. Furthermore, if the weights and the threshold are encoded in unary (or are bounded polynomially in the number of states), then we obtain an exponential-time algorithm.

Second, from the reduction sketched above, we also obtain a doubly-expo\-nen\-tial upper bound on the necessary memory for~$\pO$, the~first such bound. In~contrast, a~certain succinct one-counter game due to Hunter~\cite{Hunter14arxiv}, which can easily be expressed as an average-energy game with threshold zero, shows that our bound is almost tight: in the resulting game of size $n$, energy level $2^{({2^{\sqrt{n}}}/{\sqrt{n}})-1}$ is necessary to win. Again, in the case of unary encodings, we obtain an (almost) tight exponential bound on the memory requirements.

Third, we improve the lower bound on the complexity of solving average-energy games with only a lower bound on the energy level from $\EXPTIME{}$ to $\EXPSPACE{}$ by a reduction from succinct one-counter games~\cite{Hun15}.

Fourth, we show that multi-dimensional average-energy games are undecidable, both for the case without any bounds and for the case of only lower bounds. Only the case of games with both lower and upper bounds turns out to be decidable: it~is shown to be both in $\NEXPTIME$ and in $\co\NEXPTIME$. This~problem trivially inherits $\EXPTIME{}$-hardness from the one-dimensional case.

\section{Preliminaries}
\label{sec-defs}
\paragraph{Graph games.} 
We consider finite turn-based weighted games played on graphs between
two players, denoted by $\pO$ and~$\pI$. Such a game is a tuple $\Game
= (S_0, S_1, \trans)$ where (i)~$S_0$~and $S_1$ are disjoint sets of \textit{states} belonging to $\pO$ and~$\pI$, with $S = S_0
\uplus S_1$, (ii)~$\trans \subseteq S \times [-W;W] \times S$, for some~$W\in\bbN$, is a set of
integer-weighted \textit{edges}.  Given an edge $e=(s, w, t) \in \trans$, we write
$\src(e)$ for the source state~$s$ of~$e$, $\tgt(e)$ for its target state~$t$,
and $\weg(e)$ for its weight~$w$. We assume that for every $s \in S$, there is at least one outgoing edge $(s,w,s')
\in E$.

Let $s\in S$.  A~\textit{play} from~$s$ is an infinite sequence of edges $\play = (e_i)_{1\leq i}$ such that $\src(e_1) = s$ and $\tgt(e_i)=\src(e_{i+1})$ for all
$i \ge 1$. A play induces a corresponding sequence of states, denoted $\hat{\play} = (s_j)_{0\leq j}$, such that for any $e_i$, $i \geq 1$, in $\play$, $s_{i-1} = \src(e_i)$ and $s_{i} = \tgt(e_i)$. We write $\first(\play)=s_0$ for its initial state (here,~$s$). A~play \textit{prefix} from~$s$ is a finite sequence of edges $\rho = (e_i)_{1\leq i \leq k}$ following the same rules and notations. We additionally write $\last(\rho)= s_k = \tgt(e_k)$ for its last state. We~let~$\epsilon_s$ (or~$\epsilon$ when $s$~is clear from the context) denote the empty play prefix from~$s$, with $\last(\epsilon_s) = \first(\epsilon_s) = s$. A~non-empty prefix $\rho$ such that $\last(\rho) = \first(\rho)$ is called a \textit{cycle}. The length of a prefix $\rho = (e_i)_{1\leq i \leq k}$ is its number of edges, i.e., $\length(\rho) = k$. For a play~$\pi$, $\length(\play) = \infty$. 
Given a prefix~$\rho$ and a play (or prefix)~$\pi$ with $\last(\rho)=\first(\pi)$,
the~concatenation between~$\rho$ and~$\pi$ is denoted by $\rho\cdot\pi$.

For~a play $\play = (e_i)_{1\leq i}$
and $1 \le j \le k$, we~write $\play_{[j,k]}$ to denote the finite
sequence $(e_i)_{j\leq i\leq k}$, which is a prefix from~$\src(e_j)$;
we write $\play_{\le k}$ for $\play_{[1,k]}$. For~any $i \geq 1$ and $j \geq 0$, we write $\play_i$ for edge $e_i$ and $\hat{\play}_j$ for state $s_j$. Similar notations are used for prefixes~$\rho$, with all indices bounded by $\length(\rho)$.

The set of all plays in~$\Game$ from a state~$s$ is denoted by
$\plays(\Game,s)$, and the set of all such prefixes is denoted by
$\fruns(\Game,s)$. We write $\plays(\Game)$ and $\fruns(\Game)$ for the unions of those sets over all states. We~say that a prefix $\prefix \in \prefs(\Game)$ belongs to
$\player{i}$, for $i \in \{0,1\}$, if $\last(\prefix) \in S_i$. The~set of
prefixes that belong to $\player{i}$ is denoted by $\fruns_i (\Game)$, and we define $\fruns_i (\Game,s) = \fruns_i (\Game) \cap \fruns (\Game,s)$.

\paragraph{Payoffs.} 
 Given a non-empty prefix $\rho =
(e_i)_{1\leq i\leq n}$, we
define the following payoffs:
\begin{itemize}
\item its \emph{energy level} as
  $
  \EL(\rho) = \sum_{i = 1}^{n}
  \weg(e_i);
  $
\smallskip
\item its \textit{mean-payoff} as
  $\MP(\rho) =
  \frac{1}{n} \sum_{i = 1}^{n} \weg(e_i) = \frac{1}{n} \EL(\rho);$
\smallskip
\item its \textit{average-energy} as
  $
    \AEpay(\rho) = \frac{1}{n} \sum_{i = 1}^{n} \EL(\rho_{\le i}).
  $
\end{itemize}
These definitions are extended to plays by taking the upper limit
of the respective functions applied to the sequence of
prefixes of the plays, e.g.,
\[  
    \AEsup(\play) = 
     \limsup\nolimits_{n\to\infty}\frac{1}{n} \sum\nolimits_{i = 1}^{n} \EL(\play_{\le i}).
\]

\begin{example}
  We~illustrate those definitions on a small example depicted in Fig.~\ref{fig-ex}: it displays two small (1-player, deterministic) weighted games, together with the evolution of the energy level and average-energy along their unique play. As noted in~\cite{BMRLL16}, the \textit{average-energy} can help in discriminating plays that have identical \textit{total-payoffs} (i.e., the limits of high and low points in the sequence of energy levels), in the same way that total-payoff can discriminate between plays having the same \textit{mean-payoff}. Indeed, in our example, both plays have mean-payoff equal to zero and supremum (resp.~infimum) total-payoff equal to three (resp.~$-1$), but they end up having
  different averages: the~average-energy is~$1/2$ for the left play, while it is~$3/2$ for the right one.
\end{example}

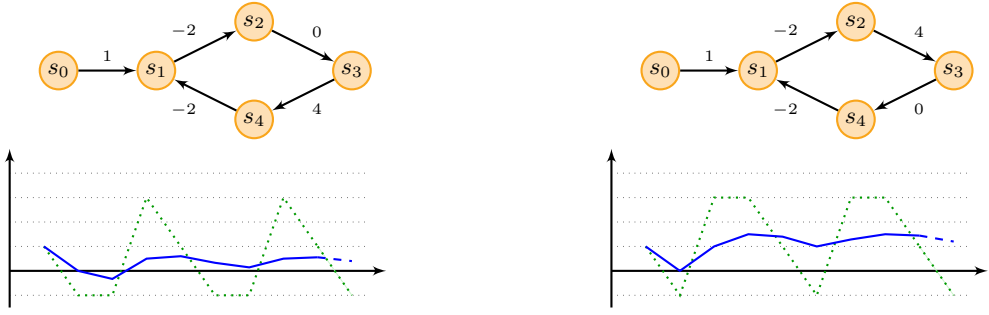
\begin{figure}[t]
\centering
\begin{tikzpicture}[thick]
\begin{scope}[scale=1.3]
\draw (0,-.25) node[medrond,jaune] (a) {} node {$s_0$};
\draw (1,-.25) node[medrond,jaune] (b) {} node {$s_1$};
\draw (2,-.75) node[medrond,jaune] (d) {} node {$s_4$};
\draw (3,-.25) node[medrond,jaune] (e) {} node {$s_3$};
\draw (2,.25) node[medrond,jaune] (c) {} node {$s_2$};
\draw[-latex'] (a) -- (b) node[midway,above] {$\scriptstyle 1$};
\draw[-latex'] (b) -- (c) node[midway,above left] {$\scriptstyle -2$};
\draw[-latex'] (c) -- (e) node[midway,above right] {$\scriptstyle 0$};
\draw[-latex'] (e) -- (d) node[midway,below right] {$\scriptstyle 4$};
\draw[-latex'] (d) -- (b) node[midway,below left] {$\scriptstyle -2$};
\begin{scope}[yshift=-2.3cm,xshift=-.5cm,xscale=.35,yscale=.25]
\draw[-latex'] (0,0) -- (11,0);
\draw[-latex'] (0,-1.5) -- (0,5);
\foreach \y in {-1,1,2,3,4} 
  {\draw[dotted,line width=.2pt] (0,\y) -- +(10.5,0);}
\draw[vert,fill=none,line width=.8pt,dotted] (1,1) -- (2,-1) -- (3,-1) --
  (4,3) -- (5,1) -- (6,-1) -- (7,-1) -- (8,3) -- (9,1);
\draw[dashed,vert,line width=.8pt,dotted] (9,1) -- (10,-1);
\draw[blue,line width=.8pt] (1,1) -- (2,0) -- (3,-.33) -- (4,.5) --
  (5,.6) -- (6,.33) -- (7,.143) -- (8,.5) -- (9,.556);
\draw[blue,dashed,line width=.8pt] (9,.556) -- (10,.4);
\end{scope}
\end{scope}
\begin{scope}[xshift=8cm,scale=1.3]
\draw (0,-.25) node[medrond,jaune] (a) {} node {$s_0$};
\draw (1,-.25) node[medrond,jaune] (b) {} node {$s_1$};
\draw (2,-.75) node[medrond,jaune] (d) {} node {$s_4$};
\draw (3,-.25) node[medrond,jaune] (e) {} node {$s_3$};
\path[use as bounding box] (0,0);
\draw (2,.25) node[medrond,jaune] (c) {} node {$s_2$};
\draw[-latex'] (a) -- (b) node[midway,above] {$\scriptstyle 1$};
\draw[-latex'] (b) -- (c) node[midway,above left] {$\scriptstyle -2$};
\draw[-latex'] (c) -- (e) node[midway,above right] {$\scriptstyle 4$};
\draw[-latex'] (e) -- (d) node[midway,below right] {$\scriptstyle 0$};
\draw[-latex'] (d) -- (b) node[midway,below left] {$\scriptstyle -2$};
\begin{scope}[yshift=-2.3cm,xshift=-.5cm,xscale=.35,yscale=.25]
\draw[-latex'] (0,0) -- (11,0);
\draw[-latex'] (0,-1.5) -- (0,5);
\foreach \y in {-1,1,2,3,4} 
  {\draw[dotted,line width=.2pt] (0,\y) -- +(10.5,0);}
\draw[vert,fill=none,line width=.8pt,dotted] (1,1) -- (2,-1) -- (3,3) --
  (4,3) -- (5,1) -- (6,-1) -- (7,3) -- (8,3) -- (9,1);
\draw[dashed,vert,line width=.8pt,dotted] (9,1) -- (10,-1);
\draw[blue,line width=.8pt] (1,1) -- (2,0) -- (3,1) -- (4,1.5) --
  (5,1.4) -- (6,1) -- (7,1.286) -- (8,1.5) -- (9,1.444);
\draw[blue,dashed,line width=.8pt] (9,1.444) -- (10,1.2);
\end{scope}\end{scope}
\end{tikzpicture}
\caption{Two plays with identical mean-payoffs and total-payoffs. The left one has average-energy $1/2$, in contrast to $3/2$ for the right one. Green (dotted) and blue curves respectively represent the energy level and the average-energy over prefixes.}
\label{fig-ex}
  \end{figure}

\paragraph{Strategies.} 
A~\textit{strategy} for $\player{i}$, with $i \in \{0,1\}$, from a state $s$ is a function $\St_i\colon \fruns_i(\Game, s) \to \trans$ satisfying
${\src(\St_i(\prefix)) = \last(\prefix)}$ for all
$\prefix \in \fruns_i(\Game, s)$. We denote by $\strats_{i}(\Game, s)$, the set of strategies for~$\player{i}$ from state~$s$. We~drop $\Game$ and~$s$ when they are clear from the context.

A play $\play = (e_j)_{1\leq j}$ from $s$ is called an \textit{outcome} of strategy~$\St_i$
of~$\player{i}$ if, for all $k \ge 0$ where $\pi_{\le k} \in
\fruns_i(\Game, s)$, we have $\St_i(\play_{\le k}) = e_{k+1}$. Given a state
$s \in \states$ and strategies $\St_0$ and $\St_1$ from $s$ for both
players, we denote by $\out(s, \St_0,\St_1)$ the unique 
play that starts in $s$ and is an outcome of both $\St_0$
and~$\St_1$. When fixing the strategy of only~$\player{i}$, we denote
the set of outcomes by 
\[
\outs(s, \St_i) = 
 \{\outs(s, \St_0,\St_1) \mid \St_{1-i}\in\strats_{1-i}(\Game, s)\}.
 \]

\paragraph{Objectives.} 
An~objective in $\Game$ is a set $\mathcal{W} \subseteq
\plays(\Game)$.  Given a game $\Game$, an~initial state~$\initState$,
and an objective~$\mathcal{W}$, a~strategy $\St_0 \in \strats_{0}$ is
winning for $\pO$ if $\out(\initState, \St_0) \subseteq \mathcal{W}$.

We consider the following objectives for~$P_0$:
\begin{itemize}
\item The 
  \textbf{lower-bounded energy} objective
  $\LBound = \{ \play \in \plays(G)$ $\mid \forall\, n
  \ge 1,\penalty100\ \EL(\play_{\le n}) \geq 0 \}$ requires
a  non-negative energy level at all times.\footnote{For the sake of readability, we assume the initial credit to be zero for energy objectives throughout this paper. Still, our techniques can easily be generalized to an arbitrary initial credit~$\initCredit \in \bbN$.}
\item Given an upper bound $U \in \bbN$, the \textbf{lower- and upper-bounded energy} objective $\LUBound(U) = \{ \play \in
  \plays(\Game) \mid \forall\, n \ge 1,\ \EL(\play_{\le
    n}) \in [0,U] \}$ requires that the energy always remains
  non-negative and below the upper bound $U$ along a play.
\item Given a threshold $t \in \bbQ$, the \textbf{mean-payoff} objective $\MeanPayOff(t) = \{ \play \in \plays(\Game) \mid
  \MPsup(\play) \le t \}$ requires that the mean-payoff is at
  most~$t$.
\item Given a threshold $t \in \bbQ$, the \textbf{average-energy} objective $\AvgEnergyLevel(t) = \{ \play \in
  \plays(\Game) \mid \AEsup(\play) \le t \}$ requires
  that the average-energy is at most~$t$.
\end{itemize}
For the $\MeanPayOff$ and $\AvgEnergyLevel$ objectives, $\pO$ aims to
\textit{minimize} the payoff.

\paragraph{Decision problem.} 
In this paper, we focus on weighted games with a combination of energy and
average-energy objectives, by a detour via mean-payoff objectives. 
The exact problem we tackle is named the \textit{AEL threshold problem} and is defined as follows: given a finite game~$\Game$, an~initial state $\initState \in
\states$, and a threshold $t
\in \bbQ$ given as a fraction $\frac{t_1}{t_2}$ with $t_1$ and $t_2$ two natural numbers given in binary, decide whether~$\pO$ has a winning
strategy from~$\initState$ for the objective~$\AvgLower(t) =\LBound \cap
\AvgEnergyLevel(t)$.
As for the threshold, we consider a binary encoding of the weights in $\Game$: we thus study the complexity of the problem with regard to the length of the input's binary encoding (i.e., the number of bits used to represent the graph and the numbers involved).

Variants of this problem involving the above-mentioned payoff
functions, and combinations thereof, had been previously investigated,
see Table~\ref{tab:results} for a summary of the results. In this paper, we focus on
the remaining case, namely 2-player games with \textit{AEL} objectives, for
which decidability was not known, and proving the computational- and
memory complexities given in the corresponding cells of the table.

\renewcommand{\arraystretch}{1.1}
\begin{table*}[t]\centering
\scalebox{1}{\begin{tabular}{cccccc}\toprule
Game objective & \textbf{1-player} && \textbf{2-player}&& \textbf{memory}  \\
 \midrule

\textit{MP}  & \PTIME~\cite{Kar78} && \NP $\cap$ \coNP~\cite{ZP96} && memoryless~\cite{EM79} \\
\textit{EGL}  & \PTIME \cite{bouyer2008} && \NP $\cap$ \coNP \cite{bouyer2008,emsoft2003-CAHS} && memoryless~\cite{emsoft2003-CAHS} \\
\textit{EGLU} & \PSPACE-c. \cite{FJ13} && \EXPTIME-c. \cite{bouyer2008} &&
exponential \cite{BMRLL16} \\
 \midrule

\textit{AE}  & \PTIME \cite{BMRLL16} && \NP $\cap$ \coNP
\cite{BMRLL16}  && memoryless \cite{BMRLL16}  \\
\textit{AELU} & \PSPACE-c. \cite{BMRLL16}  &&
\EXPTIME-c. \cite{BMRLL16}  && exponential \cite{BMRLL16}  \\
\textit{AEL}  & \PSPACE-e.~/~\NP-h. \cite{BMRLL16}  && {\color{red}
\EXPTIME[2]-e.~/~\EXPSPACE-h.} && \textcolor{red}{super-exp.~/~doubly-exp.} \\
\bottomrule
\end{tabular}}
\vspace{1mm}
\caption{Complexity of deciding the winner and memory requirements for
  quantitative games: \textit{MP}~stands for mean-payoff, \textit{EGL} (resp.~\textit{EGLU}) for lower-bounded (resp.~lower- and
  upper-bounded) energy, \textit{AE} for average-energy,
  \textit{AEL} (resp.~\textit{AELU}) for average-energy under a lower
  bound (resp.~and upper bound $U \in \mathbb{N}$) on the energy, c.~for
  complete, e.~for easy, and h.~for hard. All memory bounds are tight (except for \textit{AEL}).}
\label{tab:results}
\end{table*}

\section{Equivalence with an Infinite-State Mean-payoff Game}
\label{subsec:reductionToMP}
Let $G=(S_0,S_1,E)$ be a finite weighted game, $\initState \in S$ be
an initial state, and $t
\in \bbQ$ be a threshold.  We define its \emph{expanded infinite-state
weighted game} as $G'=(\Gamma_0,\Gamma_1,\Delta)$ defined by
\begin{itemize}
\item $\Gamma_0 = S_0 \times \bbN$, and $\Gamma_1 = S_1 \times \bbN
  \uplus \{\bot\}$ (where
  $\bot$ is a fresh sink state that does not belong to~$G$); then $\Gamma =
  \Gamma_0 \uplus \Gamma_1$ is the global set of states;
\item $\Delta$~is composed of the following edges:
  \begin{itemize}
  \item a transition $((s,c),c',(s',c')) \in \Delta$
    whenever there is $(s,w,s') \in E$ with $c' = c+w \ge 0$;
  \item a transition $((s,c),\ceil{t}+1,\bot) \in \Delta$ whenever there is $(s,w,s') \in E$
    such that $c+w <0$;
  \item finally, a transition $(\bot,\ceil{t}+1, \bot) \in \Delta$.
  \end{itemize}
\end{itemize}
In this expanded game, elements of $\Gamma$
are called \emph{configurations}, and the initial configuration is set
to $(\initState,0)$.

\begin{restatable}{lemma}{lemReducToMP}
\label{lem:reducToMP}
  Player $\pO$ has a winning strategy in~$G$ from state~$\initState$ for the objective
  $\AvgLower(t)$ if, and only~if, he has a winning strategy in~$G'$ from
  configuration $(\initState,0)$ for the objective $\MeanPayOff(t)$.
\end{restatable}

\begin{proof}  
  Let $\iota$ be the mapping that assigns to every play $\pi =
  (e_i)_{1\leq i}$ from~$\initState$ an expanded play in~$G'$ from
  $(\initState,0)$, defined inductively as follows.
	\begin{itemize}
		\itemsep0.3em
	\item $\iota(\epsilon) = \epsilon$;
	\item 
	Let $\gamma_n=\last(\iota(\pi_{\leq n}))$ for $n \geq 0$ with $\gamma_0 = last(\iota(\epsilon)) = (\initState, 0)$, the initial configuration. For $n \geq 1$, $\pi_{\leq n} = \pi_{\leq n - 1} \cdot (s_{n-1}, w_n, s_{n})$, we define $\iota(\pi_{\leq n}) =\iota(\pi_{\leq n-1}) \cdot \delta_n$ where $\delta_n \in \Delta$ is defined as follows:
	\vspace{0.3em}
		\begin{itemize}
		\itemsep0.3em
	  		\item if $\gamma_{n -1} = \bot$ then $\delta_n=(\bot, \ceil{t}+1, \bot)$;
			\item else $\gamma_{n-1} = (s_{n-1}, c_{n-1})$ and if $c_{n-1} + w_n < 0$ then $\delta_n=(\gamma_{n-1},\ceil{t}+1, \bot)$;
			\item else $\delta_n=(\gamma_{n-1}, c_n, \gamma_{n})$ for $c_n = c_{n-1} + w_n$ and $\gamma_{n}= (s_{n},c_n)$.
		\end{itemize}
	\end{itemize}

  One easily checks that the following properties hold:
  \begin{itemize}
  \item for any $n \geq 0$, $\length(\iota(\pi_{\leq n})) = \length(\pi_{\leq n})$,
  \item $\iota(\pi)$ reaches $\bot$ if, and only if, $\EL(\pi_{\leq j})<0$ for some~$j\geq 1$,
  \item if $\iota(\pi)$ never reaches $\bot$, then for all $j \geq 1$,
   \[
  \MP(\iota(\pi_{\leq j})) = \AEpay(\pi_{\leq j}).
  \]
\end{itemize}   
Hence, if a play $\pi$ in $G$ always keeps its energy level non-negative, we have that $\MPsup(\iota(\pi)) = \AEsup(\pi)$, and if it does not, that $\MPsup(\iota(\pi)) = \ceil{t}+1$ thanks to the sink state $\bot$. Furthermore, observe that if a play $\pi'$ in $\Game'$ does not reach $\bot$, then there is a unique corresponding play $\pi = \iota^{-1}(\pi')$ in $\Game$: i.e., we have a bijection for plays that avoid~$\bot$.
  
  Now, pick a winning strategy~$\sigma_0$ for~$\pO$ in~$G$ from
  $\initState$ for $\AvgLower(t)$. First notice that
  for every $\pi \in \outs(\initState,\sigma_0)$, for every $j \ge 1$,
  $\EL(\pi_{ \le j}) \ge 0$; in particular, $\iota(\pi)$
  never visits~$\bot$ in~$G'$. We define the strategy~$\sigma'_0$ in~$G'$ that mimics~$\sigma_0$ based on the aforementioned bijection. By construction, we have that any outcome $\pi'$ of $\sigma'_0$ avoids $\bot$ and thus that its mean-payoff is equal to the average-energy of $\iota^{-1}(\pi')$ in $G$, which is an outcome of $\sigma_0$. Hence, by hypothesis on $\sigma_0$, $\sigma'_0$ is winning for $\MeanPayOff(t)$.
  
  Conversely, pick a winning strategy~$\sigma'_0$ in~$G'$ from
  $(\initState,0)$ for objective
  $\MeanPayOff(t)$. No~outcome may visit $\bot$, otherwise the mean-payoff would be $\ceil t+1>t$. Hence the bijection holds for all outcomes of $\sigma'_0$ and we use it to define the strategy $\sigma_0$ that mimics $\sigma'_0$ in $G$. By hypothesis on $\sigma'_0$, any outcome $\pi$ of $\sigma_0$ will keep the energy level non-negative at all times, and will be such that $\AEsup(\pi) = \MPsup(\iota(\pi))$, hence will have its average-energy bounded by $t$ since $\iota(\pi)$ is an outcome of $\sigma'_0$. Thus $\sigma_0$ is winning for $\AvgLower(t)$, which concludes our proof.\qed
\end{proof}

For the rest of this paper, we fix a weighted game $G = (S_0,S_1,E)$ and a threshold $t \in \bbQ$, and work on the corresponding expanded weighted game $G' = (\Gamma_0,\Gamma_1,\Delta)$. We~write $t = \frac{t_1}{t_2} = \Int{t} + \frac{t'}{t_2}$,
where $t_1, t_2, t' \in \bbN$ (recall they are given in binary), and $0 \le t' < t_2$, and $\Int{t}$ stands for
the integral part of~$t$.
We~also let $\tildet = \Int{t}+1-t = 1-\frac{t'}{t_2}$. Hence $\tildet=1$
indicates that $t$~is an integer. For a given
threshold~$t\in \bbQ$, we consider $\Gamma^{\le t} = \{(s,c) \in
\Gamma \mid c \le t\}$, i.e., the set of configurations below the
threshold.

Note that $G'$ can be interpreted as a one-counter pushdown mean-payoff game with an unbounded weight function. While it is
well-known how to solve mean-payoff games over \textit{finite} arenas,
not much is known for infinite arenas (see
Section~\ref{sec-intro}). However, our game has a special structure
that we will exploit to obtain an algorithm. Roughly, our approach consists in transforming the $\AvgLower(t)$
objective into an equivalent $\AvgLowerUpper(t,U) = \LUBound(U) \cap
  \AvgEnergyLevel(t)$ objective, where (the value
of)~$U$ is doubly-exponential in the input by analyzing plays and strategies in $G'$. In other terms, we show
that any winning strategy for $\AvgLower(t)$ can be transformed into
another winning strategy along which the energy level remains
upper-bounded by~$U$. 

The proof is two-fold: we~first show (in Section~\ref{sec-bounding}) that we can
bound the energy level for the special case where the objective
consists in reaching a finite set of configurations of the game (with only a
lower bound on the energy level). This is achieved by a detour to
pushdown games: while there are known algorithms for solving
reachability pushdown games, to the best of our knowledge, there are no (explicit) results bounding the maximal
stack height.

As a second step (in Section~\ref{sec-algo}), we~identify \emph{good cycles}
in winning outcomes, and prove that they can be shown to have bounded length. The
initial configurations of those cycles will then be the targets of the
reachability games above. Combining these two results yields the desired upper bound on the energy levels.

\section{Bounding One-counter Reachability Games}
\label{sec-bounding}
We~focus here on a reachability objective in~$G'$, where the
target set is a subset~$\Gamma'\subseteq \Gamma^{\le t}$: we~aim at
bounding the maximal energy level that needs to be visited with a winning
strategy.

The game $G'$ is a particular case of a pushdown game~\cite{Walukiewicz01}. Hence we 
use results on pushdown games, and build a new winning strategy, in
which we will be able to bound the energy level at every visited
configuration. Note that the bound~$M'$ in the following lemma is doubly-exponential,
as we encode~$W$, the~largest absolute weight in~$G$, and the
threshold~$t$, in binary. The proof of the next lemma is based
on the reformulation of the algorithm from~\cite{Walukiewicz01} made
in~\cite{FridmanZimmermann12}.

\begin{restatable}{lemma}{lemPushdownGames}
  \label{lemma:pushdown_games}
  Fix $M \in \mathbb{N}$. There exists $M' = 2^{\mathcal{O}(M +
    \size{S} + \size{E} \cdot W + \size{S} \cdot (\lceil t\rceil
    +1))}$ such that for every $\gamma=(s,c)$ with $c \le M$ and for
  every $\Gamma' \subseteq \Gamma^{\le t}$, if there is a strategy for
  $\pO$ to reach $\Gamma'$ from $\gamma$ in $G'$, then there is also a
  strategy which ensures reaching $\Gamma'$ from $\gamma$ without
  exceeding energy level $M'$.
\end{restatable}

\begin{proof}
\newcommand{\pdg}{\mathcal{P}}
\newcommand{\fpg}{\mathcal{G}}

We rely on Walukiewicz's~\cite{Walukiewicz01} solution of parity games on configuration graphs of pushdown machines via a simulation by parity games on finite graphs as presented in \cite{FridmanZimmermann12}. First, we construct such a pushdown parity game~$\pdg$ that is won by $\pO$ if, and only if, there is a strategy for $\pO$ to reach $\Gamma'$ from $\gamma$ in $G'$. Then, we apply Walukiewicz's  reduction, which yields a winning strategy~$\sigma_\pdg$ for $\pdg$ that is induced by a memoryless winning strategy~$\sigma_\fpg$ for  the parity game $\fpg$ constructed in the reduction. Finally, we show that this strategy can be turned into a strategy for~$G'$ that ensures reaching~$\Gamma'$ without exceeding energy level~$M'$. 

When interpreting $G'$ as a pushdown game, the energy level is stored in unary on the stack, and the stack alphabet has only two symbols, one for the counter, and one (denoted by $\bot$) for the bottom of the stack. Our weights here are encoded in binary, which implies that an exponential number of stack symbols may be pushed onto or may be popped of the stack during a single transition. On the other hand, Walukiewicz's construction considers pushdown machines where at most one symbol may be pushed or popped during a single transition. Obviously, pushing and popping multiple symbols can be simulated in the latter model by adding additional states, at the price of an exponential blow-up. 

We begin with the definition of $\pdg$, which runs in three phases. In the first phase, the configuration~$\gamma$ is produced. This phase is deterministic, i.e., the players have no choices in this phase. Then, in the second phase, a play of $G'$ starting in $\gamma$ is simulated. At any point of this simulation, $\pO$ can stop the simulation and start a third phase. Say, he stops the simulation of $G'$ in configuration~$\gamma'$. The third phase is designed to check whether $\gamma'$ is in $\Gamma'$. If this is the case, then the configuration~$(q_f, \bot)$ is reached, where $q_f$ is a fresh state. If not, then a configuration~$(q_r, \bot)$ is reached, where $q_r$ is again a fresh state. This is implemented by deterministically determining the stack height of the configuration~$\gamma'$, which can be stopped as soon as it turns out that the stack height is greater than $t$, as this implies $\gamma' \notin \Gamma'$. If $\pO$ does not stop the simulation, the configuration~$(q_f, \bot)$ is never reached.

Thus, the number of states of the pushdown system underlying $\pdg$ can be bounded by $\mathcal{O}(M + \size{S} + \size{E} \cdot  W + \size{S} \cdot  (\lceil t\rceil +1))$, $M+1$ states for the first phase, $\mathcal{O}(\size{S} + \size{E} \cdot  W)$ for the second one, and $\size{S} \cdot  (\lceil t\rceil +1) +2$ states for the third phase.

By construction, the premise of the statement we prove, i.e., there is a strategy for $\pO$  to reach $\Gamma'$ from $\gamma$ in $G'$, is equivalent to $\pO$ having a strategy for $\pdg$ to reach the configuration~$(q_f, \bot) $ from the initial configuration, which is a reachability objective. We turn this into a parity objective as usual: All states are colored by $1$ save for a fresh one with color~$0$ that is only reachable from the configuration~$(q_f, \bot)$. This fresh state is equipped with a self-loop, so that the parity objective is satisfied if, and only if, the fresh state is reached via $(q_f, \bot)$. 

Using the adapted construction of Walukiewicz as presented in \cite{FridmanZimmermann12}, we obtain a parity game~$\fpg$ on a finite graph of size $2^{\mathcal{O}(\size{\pdg})}$ with a designated initial state $s_0$, and such that $0$ is the only even color in $\fpg$ (the larger harmless color needed in \cite{FridmanZimmermann12} for auxiliary vertices can be taken as $3$). The crucial property of $\fpg$ is that $\pO$ has a winning strategy in $\fpg$ from $s_0$ if, and only if, $\pO$ has a winning strategy in $\pdg$ from its initial configuration (that we denote from now on by $(q_0,\bot)$). Even more so, a memoryless winning strategy~$\sigma_\fpg$ for $\pO$ in $\fpg$ from $s_0$ can be turned into a winning strategy~$\sigma_\pdg$ for $\pO$ in $\pdg$ from $(q_0,\bot)$ so that: for every play prefix~$\rho \in \outs((q_0,\bot), \sigma_{\pdg})$ which has not reached yet a configuration of color~$0$, there is a play prefix~$\rho' \in \outs(s_0,\sigma_{\fpg})$ which does not contain any vertex of color~$0$, such that the stack height of $\rho$ corresponds to the number of vertices of color~$1$ occurring in $\rho'$.\footnote{The result proven in \cite{FridmanZimmermann12} is actually more general. The statement here follows from the fact that the stack height of $\rho$ is equal to the stair-score of color~$1$ in the setting we consider here.}

Now, assume such a play prefix $\rho$ in $\pdg$ reaches stack height greater than $\size{\fpg}$. Then, the corresponding play prefix~$\rho'$ in $\fpg$ visits at least one vertex of color~$1$ twice without an occurrence of color~$0$ in between. From such a play prefix, since $\sigma_\fpg$ is memoryless, one can construct an infinite outcome of $\sigma_\fpg$ from $s_0$, which visits vertices of color~$1$ infinitely often, but no vertex of color~$0$. Such a play is losing for $\pO$, which contradicts the fact that $\sigma_\fpg$ is winning from $s_0$. Hence, such a play prefix~$\rho$ does not exist, i.e., the stack height of plays that start in the initial configuration, that are consistent with $\sigma_\pdg$, and have not yet reached $(q_f, \bot)$ is bounded by $\size{\fpg}$. 

Finally, the strategy~$\sigma_\pdg$ can now easily be transformed into a strategy for $G'$ that ensures 
  reaching $\Gamma'$ from $\gamma$ without exceeding energy level~$M' = \size{\fpg} = 2^{\mathcal{O}(M + \size{S} + \size{E} \cdot  W + \size{S} \cdot  (\lceil t\rceil +1))}$.\qed
\end{proof}

\section{A Doubly-exponential Time Algorithm}
\label{sec-algo}
Let $\rho = (e_i)_{1\leq i\leq n}$ be a prefix in~$G'$ such that its configurations are denoted $(s_0, c_0), \ldots{}, (s_n, c_n)$. 
Observe that by construction of the expanded game, $\MP(\rho) = \frac{1}{n}\sum_{i=1}^n c_n$. 
In particular, note that the initial configuration $(s_0, c_0)$ does not appear when computing its mean-payoff.

Let
$\widetilde\Gamma \subseteq \Gamma$ be a
set of configurations of~$G'$ and $\rho$ a play prefix. We~define $\density(\widetilde\Gamma,\rho)$~by:
\[
\density(\widetilde\Gamma,\rho) = \frac{|\{1 \le i \le \length(\rho) \mid \hat{\rho}_i
  \in \widetilde\Gamma\}|}{\length(\rho)},
\] 
which denotes the proportion (or~\emph{density}) of configurations belonging
to $\widetilde\Gamma$ along~$\rho$. 
Observe that the initial configuration $\hat{\rho}_0$ is not
taken into account: this is because $\density(\widetilde\Gamma,\rho)$
will be strongly linked to the mean-payoff, as we now explain.

\subsection{Analyzing Winning Plays}
\label{ssec-winplays}

In this section, we~analyze winning plays in~$G'$, and prove that they
must contain a cycle that is ``short enough'' and has mean-payoff less
than or equal to~$t$.
To achieve this, we first establish a lower bound on the density of configurations below the threshold along a winning play.

\begin{restatable}{lemma}{lemmadensity}
  \label{lemma-density}
  Let~$\pi$ be a play in~$G'$ from $(\initState,0)$ with $\MPsup(\pi)\leq t$. Then, there
  exists $n \in \bbN$ such that for every $n' \ge n$,
  \[
  \density(\Gamma^{\le t},\pi_{\leq n'}) \ge \frac{\tildet}{2(t +1)}.
  \]
\end{restatable}

\begin{proof}
  Write $\pi = (e_i)_{1\leq i}$, and write $\hat{\pi}_j=(s_j,c_j)$ for all $j\geq 0$. That $\MPsup(\pi)\leq t$ means that for any $\varepsilon>0$, there exists $n \in
  \bbN$ such that for every $n' \ge n$,
  \[
  \frac{\sum_{i=1}^{n'} c_i}{n'} \le t +\varepsilon.
  \]
  Pick $0<\varepsilon\leq \frac{\tildet}{2}$, and $n\in\bbN$ for which the above property holds.
  For any $n'\geq n$, let~$I$ be the subset of positions~$i$ in
  $[1,n']$ such that $c_i \le t$, and $J$ its complement: $I \uplus J
  = [1,n']$. We can then compute:
  \begin{eqnarray*}
    t+\varepsilon \geq \frac{\sum_{i=1}^{n'} c_i}{n'} &= &\frac{\sum_{i \in I} c_i +
      \sum_{i \in J} c_i}{ n'} \\
    & \ge & 0 + \frac{\sum_{i \in J} (\Int{t} + 1)}{n'} = 
    \frac{|J| \cdot (\Int{t} +1)}{n'}
  \end{eqnarray*}
  since any $c_i \in \bbN$ that is strictly greater than~$t$ is larger
  than or equal to $\Int{t} +1$.  Hence $t +\varepsilon \ge \frac{|J|
    \cdot(\Int{t} +1)}{n'}$ or equivalently
  \begin{equation}\label{eq:boundOnJ}
    \frac{t+\varepsilon}{\Int{t}+1}\ge \frac{|J|}{n'}.
  \end{equation}
  Now:
  \begin{xalignat*}1
    \density(\Gamma^{\le t}, \pi_{\leq n'}) &=  \frac{|I|}{n'} =
    1-\frac{|J|}{n'} \\
    & \ge  1-\frac{t+\varepsilon}{\Int{t}+1} \tag{due to equation \ref{eq:boundOnJ}}\\
    &\ge  \frac{\Int{t} +1-t-\varepsilon}{\Int{t}+1} = \frac{\tildet-\varepsilon}{\Int{t}+1}\\
    &\ge  \frac{\tildet}{2(\Int{t} +1)} \tag{since
    $\varepsilon \le \frac{\tildet}2$} \\
    &\ge  \frac{\tildet}{2(t +1)}. 
  \end{xalignat*}
  This concludes the proof of the lemma. \qed
\end{proof}

We now deduce that there is at least one configuration below the threshold that appears with lower-bounded density.

\begin{restatable}{lemma}{corodensity}
\label{coro:density}
Let~$\pi$ be a play in~$G'$ from $(\initState,0)$ with
$\MPsup(\pi)\leq t$.  There exists $\gamma \in \Gamma^{\le t}$ such
that for any~$n\in\bbN$, there exists infinitely many
positions~$n'\geq n$ for which
  \[
  \density(\{\gamma\},\pi_{[n,n']}) \ge \frac{\tildet}{4(t +1)^2 |S|}.
  \]
\end{restatable}

\begin{proof}
Since $\size{\Gamma^{\le t}} \leq (t+1)\cdot |S|$,\footnote{This is thanks to the special structure of our infinite-state game.} we easily get from
Lemma~\ref{lemma-density} that there is some~$n\in\bbN$ and some configuration~$\gamma \in \Gamma^{\le t}$
such that for any~$n'\geq n$, $\density(\{\gamma\},\pi_{\leq n'}) \ge \frac{\tildet}{2(t +1)\size{\Gamma^{\le t}}} = \frac{\tildet}{2(t +1)^2
  |S|}$.

For this $n$ and such a configuration~$\gamma$, we have that for any $n' > n$:
\begin{xalignat*}1
  \density(\{\gamma\}, \pi_{\leq n'}) & = \frac{
    \size{\{1\leq i\leq n \mid (s_i,c_i)=\gamma\}} +
    \size{\{n+1\leq i\leq n' \mid (s_i,c_i)=\gamma\}}
  }{n'} \\
  & \leq \frac{n}{n'} + \frac{\size{\{n+1\leq i\leq n' \mid
      (s_i,c_i)=\gamma\}}}{n'}\\
  & \leq \frac{n}{n'} + \frac{\size{\{n+1\leq i\leq n' \mid
      (s_i,c_i)=\gamma\}}}{n'-n}\\
  &= \frac{n}{n'} + \density(\{\gamma\}, \pi_{[n,n']})
\end{xalignat*}
Thus $\density(\{\gamma\}, \pi_{[n,n']}) \geq
\frac{\tildet}{2(t+1)^2\size S} - \frac{n}{n'}$ for infinitely many~$n'$. For those $n'$ larger than $\frac{4n(t+1)^2\size
  S}{\tildet}$ (still infinitely many), it~holds
$\density(\{\gamma\}, \pi_{[n,n']}) \geq
\frac{\tildet}{4(t+1)^2\size S}$. Observe that this statement uses the $n$ from Lemma~\ref{lemma-density}. Nonetheless, it follows that it holds for any $n \in \bbN$ otherwise it would not hold either for the one given by Lemma~\ref{lemma-density} (as it needs to hold for infinitely many $n' > n$).
\qed
\end{proof}

We make a simple observation regarding cycles that have a mean-payoff exceeding the threshold.

\begin{restatable}{lemma}{lemmanotgood}
  \label{lemma:not_good}
  Consider a cycle $\rho = (e_i)_{1\leq i\leq n}$ in~$G'$, and write
  $\hat{\rho}_j = (s_j, c_j)$  for $0\leq j\leq n$. Then:
  \[
  \MP(\rho)> t =\frac{t_1}{t_2} \quad\Rightarrow\quad
  \sum_{i=1}^{n} c_i \geq \frac{t_1  n +1}{t_2}.
  \]
\end{restatable}

\begin{proof}
  Indeed, $\MP(\rho)>\frac{t_1}{t_2}$ implies that $
  t_2\cdot\sum_{i=1}^n c_i > t_1\cdot n$. Since both sides are
  integers, it~must be $t_2\cdot\sum_{i=1}^n c_i \geq t_1\cdot n +1 $,
  which entails the result.  \qed
\end{proof}

Finally, we define  the notion of \emph{good cycle}, i.e., a cycle~$\rho= (e_i)_{1\leq i\leq n}$ such that $c_0\leq t$ and $\MP(\rho)
\le t$. That is, both the initial configuration and the mean-payoff of the cycle are below the threshold. Using the previous results, we prove that any winning play must contain a good cycle.

\begin{restatable}{lemma}{lemmagoodcycle}
  \label{lemma:good_cycle}
  Let~$\pi$ be a play in~$G'$ from $(\initState,0)$ with $\MPsup(\pi)\leq t$. Then there exist $1 \le i \le j$ such that
  $\pi_{[i,j]}$ is a good cycle.
\end{restatable}

\begin{proof}
  Pick $\gamma = (s,c)$ a frequent low-energy configuration as given by
  Lemma~\ref{coro:density}. First notice that for $n$ and $n'$ two indices given by Lemma~\ref{coro:density},
  $\density(\{\gamma\},\pi_{[n, 
    n'-1]})\ge \density(\{\gamma\},\pi_{[n,n']})$ if
  $\hat{\pi}_{n'}\not=\gamma$. Hence we may reinforce
  Lemma~\ref{coro:density} by stating that $\density(\{\gamma\},
  \pi_{[n,n']})\ge \frac {\tildet}{4(t+1)^2\size S}$ for infinitely
  many $n'$ such that $\hat{\pi}_{n'}=\gamma$. We can thus write $\pi = \rho \cdot \mathcal{C}_1 \cdot \mathcal{C}_2 \dots$ where
  $\last(\rho) = \gamma$ and for all~$i\geq 1$, $\first(\mathcal{C}_i) = \last(\mathcal{C}_i) = \gamma$, and
  $\density(\{\gamma\}, \mathcal{C}_1\cdots
  \mathcal{C}_i) \ge \frac{\tildet}{4(t+1)^2\size S}$. That is, we decompose $\pi$ according to positions where $\gamma$ appears and where the density lower bound holds.

  Each~cycle $\mathcal{C}_i$ may itself contain several occurrences
  of~$\gamma$. We~can further decompose it as $\mathcal{C}_i= \mathcal{D}_i^1\cdot
  \mathcal{D}_i^1\cdots \mathcal{D}_i^{\ell_i}$ where for all~$1\leq j\leq \ell_i$, $\first(\mathcal{D}_i^j)=\last(\mathcal{D}_i^j)=\gamma$ and
  $\gamma$ does not occur anywhere else in~$\mathcal{D}_i^j$.
  
  We claim that one of these cycles $\mathcal{D}_i^j$ must be good. Toward a contradiction, assume that none is good. Then:\footnote{For the sake of readability, we sometimes write the length of a prefix $\length(\rho)$ as $\size{\rho}$ in complicated expressions.}
  \begin{xalignat*}1
    \MPsup(\pi) &= \MPsup(\pi_{> \size{\rho}}) \tag{by prefix-independence of $\MPsup$}
    \\&\ge \limsup_{k\to\infty}
    \frac{\sum_{i=1}^k\sum_{j=1}^{\ell_i} \frac{1}{t_2} (t_1 
      |\mathcal{D}_i^j|+1)}{\sum_{i=1}^k\sum_{j=1}^{\ell_i}  |\mathcal{D}_i^j|}
       \tag{by Lemma~\ref{lemma:not_good}}\\
    &=  t + \frac{1}{t_2}\limsup_{k\to\infty}
       \frac{\sum_{i=1}^k\sum_{j=1}^{\ell_i}
         1}{\sum_{i=1}^k\sum_{j=1}^{\ell_i} 
      |\mathcal{D}_i^j|} \\
    & \ge  t+\frac{1}{t_2}\limsup_{k\to\infty} \frac{\tilde{t}}{4 (t +1)^2|S|}
    \tag{by Lemma~\ref{coro:density}} \\
    & =  t + \frac{\tilde{t}}{4 t_2 (t+1)^2|S|} >t
  \end{xalignat*}
  This contradicts the fact that $\pi$ is winning in~$G'$. Hence one of the
  cycles $\mathcal{D}_i^j$ must be good.  \qed
\end{proof}

Next we establish an upper-bound on the length of the shortest good cycle that can be observed along the path.
First we prove the following lemma.

\begin{restatable}{lemma}{lemmashort}
  \label{lemma:short}
  Let $\mathcal{C}$ be a reachable good cycle with no good strict sub-cycle. Then
  $\length(\mathcal{C}) \le 8 t_1 t_2 (t +1)^3|S|^2$.
\end{restatable}

\begin{proof}
  We pick a good cycle $\mathcal{C} = (e_i)_{1\leq i\leq n}$, and
  assume that it contains no good strict sub-cycle.
  We~write $\hat{\mathcal{C}}_j = (s_j,c_j)$ for all~$0\leq j\leq n$. Let $\gamma \in \Gamma^{\le t}$ be the most frequent configuration
  of~$\Gamma^{\le t}$ in~$\hat{\mathcal{C}}$.
  Applying Lemma~\ref{coro:density} to the winning
  play~$\rho\cdot \mathcal{C}^\omega$, for some prefix~$\rho$, 
  we~easily obtain that $\density(\{\gamma\}, \mathcal{C}) \ge \frac{\tildet}{4(t
    +1)^2 |S|}$.

  Write this cycle $\mathcal{C}$ as $\rho_0 \cdot \mathcal{D}_1 \cdot \mathcal{D}_2 \dots \mathcal{D}_k
  \cdot \rho_1$, where $\gamma$ occurs in~$\rho_0$ only at the last
  position, in~$\rho_1$ only at the first position, and $\mathcal{D}_i$ are cycles
  containing~$\gamma$ only as first and last configurations. Then
  decompose $\rho_0$ and $\rho_1$ as $\rho_i^0 \cdot \mathcal{E}_i^1 \cdot \rho_i^1
  \cdot \mathcal{E}_i^2 \dots \rho_i^{\ell_i}$ such that $\rho_i^0 \cdot \rho_i^1
  \dots \rho_i^{\ell_i}$ does not contain cycles starting and ending
  in~$\Gamma^{\le t}$, and such that the $\mathcal{E}_i^j$ are cycles. This can be performed by iteratively
  ``factoring'' any maximal cycle over configurations in~$\Gamma^{\le
    t}$, until no such cycle exists.

  Let $\Omega$ be the following multiset of cycles:
  \[\Omega =
  \{\mathcal{D}_i \mid 1 \le i \le k\} \cup \{\mathcal{E}_i^j \mid
  i\in\{0,1\},\ 1 \le j
  \le \ell_i\}.\]
  We then write $I_\Omega$ for the set of positions~$1 \leq i \leq 
 \size{\mathcal{C}}$
  along~$\mathcal{C}$ such that $(s_i,c_i)$ belongs to some cycle $\omega \in
  \Omega$, and $J_\Omega = [1,n] \setminus I_\Omega$. Also, we define
  $J_\Omega^{ \bowtie t} = \{i \in J_\Omega \mid c_i \bowtie t\}$ for
  ${\bowtie} \in \{\le,>\}$. By hypothesis, the cycles~$\omega$ in~$\Omega$ are not good and are such that $(s, c) = \first(\omega)$ satisfies $c \le t$: hence it holds that $\MP(\omega)>t$ for any cycle~$\omega \in \Omega$. Thus
 \begin{xalignat*}1
  \MP(\mathcal{C}) = \frac{\sum_{i \in I_\Omega} c_i + \sum_{i \in J_\Omega} c_i}{\size{\mathcal{C}}} 
    &\ge \frac{\sum_{\omega \in \Omega} \frac{t_1 |\omega| +1}{t_2} +
      \sum_{i \in J_\Omega^{\le t}} c_i +
      \sum_{i \in J_\Omega^{>t}} c_i}{
      \sum_{\omega \in \Omega} |\omega| + |J_\Omega^{\le t}| +
      |J_\Omega^{>t}|}\\
    & \ge \frac{\sum_{\omega \in \Omega} \frac{t_1 |\omega| +1}{t_2} +
      \sum_{i \in J_\Omega^{\le t}} 0 +
      \sum_{i \in J_\Omega^{>t}} c_i}{
      \sum_{\omega \in \Omega} |\omega| + |J_\Omega^{\le t}| + |J_\Omega^{>t}|} \\
    & \ge  \frac{t \cdot \sum_{\omega \in \Omega} |\omega|+ \frac{|\Omega|}{t_2} +
      (\Int{t} +1) \cdot |J_\Omega^{>t}|}{
      \sum_{\omega \in \Omega} |\omega| + 2 \cdot |S| \cdot (\Int{t}
      +1)+ |J_\Omega^{>t}|}
  \end{xalignat*}
  where the first inequality follows from Lemma~\ref{lemma:not_good} and for the last one, we use $\sum_{i \in J_\Omega^{>t}}
  c_i\ge (\Int{t} +1) \cdot |J_\Omega^{>t}|$ (obvious) and
  $|J_\Omega^{\le t}|\le 2\cdot |S| \cdot (\Int{t} +1)$ (because
  $J_\Omega^{\le t}$ contains each configuration of~$\Gamma^{\le t}$
  at most twice, by construction of the decompositions of~$\rho_0$
  and~$\rho_1$---at~most once in each of them).  

  Now, we have that $|\Omega| \ge \density(\{\gamma\},\mathcal{C})
  \cdot |\mathcal{C}| \ge \frac{\tildet}{4(t +1)^2 |S|} \cdot |\mathcal{C}|$ and we deduce:
  \begin{eqnarray*}
    \frac{t \cdot \sum_{\omega \in \Omega} |\omega|+
      \frac{\size{\mathcal{C}}}{t_2} \frac{\tildet}{4(t+1)^2 |S|} + 
      (\Int{t} +1) \cdot |J_\Omega^{>t}|}{
      \sum_{\omega \in \Omega} |\omega| + 2\cdot |S| \cdot (\Int{t}+1)+ |J_\Omega^{>t}|} & \le &\MP(\mathcal{C}).
  \end{eqnarray*}
  Since $\MP(\mathcal{C}) \le t$, we get:
  \begin{multline*}
    t \cdot \sum_{\omega \in \Omega} |\omega|+ \frac{\size{\mathcal{C}}}{t_2}\cdot
    \frac{\tildet}{4(t+1)^2 |S|} +
    (\Int{t} +1) \cdot |J_\Omega^{>t}| \leq{} t \cdot (
    \sum_{\omega \in \Omega} |\omega| + 2\cdot |S| \cdot (\Int{t}+1)+ |J_\Omega^{>t}|)
  \end{multline*}
  and we deduce
  \begin{eqnarray*}
    \frac{\size{\mathcal{C}}}{t_2} \frac{\tildet}{4(t
      +1)^2 |S|} & \le & t \cdot 2 \cdot |S| \cdot (\Int{t} +1) -
    |J_\Omega^{>t}|     \cdot \tildet.
  \end{eqnarray*}
  From this, we obtain
  \begin{xalignat*}1
   \size{\mathcal{C}}  &\le  \frac{8 t_2 t (t +1)^3|S|^2}{\tildet} \tag{since $\tildet > 0$ and $\Int{t} + 1 \leq t+1$}\\
   &\le  \frac{8 t_1 t_2 (t +1)^3|S|^2}{t_2 - t'} \tag{using $\tildet = \frac{t_2 - t'}{t_2}$ and $t_2 t = t_1$}\\
   &\le  8 t_1 t_2 (t +1)^3|S|^2 \tag{using $t_2 - t' \geq 1$}\\
  \end{xalignat*}
This concludes the proof.  \qed
\end{proof}

For the rest of the paper, we set $N = 8 t_1 t_2 (t +1)^3|S|^2$, a bound on the size of a minimal good cycle along a winning play.

\begin{restatable}{proposition}{coroN}
  \label{coro-N}
  Let~$\pi$ be a play in~$G'$ from $(\initState,0)$ with $\MPsup(\pi)\leq t$. Then there exist $1 \le i \leq j$ such that
  $\pi_{[i,j]}$ is a good cycle of length at most $N$.
\end{restatable}

\begin{proof}
  This is a consequence of Lemma~\ref{lemma:good_cycle} and
  Lemma~\ref{lemma:short}. \qed
\end{proof}
\subsection{Strategies Described by Finite Trees}
\label{ssec-finitetrees}

So far, we have proven that $\pO$ should ``target'' short good
cycles. However in a two-player context, $\pI$ might prevent $\pO$ from doing~so.
We~therefore need to consider the
\emph{branching} behaviour of the game, and not only the linear
point-of-view given by a play. Toward that aim, we  represent
strategies (of~$\pO$) as \emph{strategy trees}, and use them to bound
the amount of memory and the counter values needed to win in~$G'$.

We consider labelled trees with backward edges $\calT =
(\calN,\calE,\lambda,\dashrightarrow)$, where $\calN$ is a finite set
of nodes, $\lambda\colon \calN \to S \times \mathbb{N}$ (each node is
labelled with a configuration of the game~$G'$), and
$\mathord{\dashrightarrow} \subseteq \calN \times \calN$. We assume
$\calT$ has at least two nodes. The relation~$\calE$ is the successor
relation between nodes.  A~node with no $\calE$-successor is called a
\emph{leaf}; other nodes are called \emph{internal nodes}.  The~root
of~$\calT$, denoted by~$\nroot$, is the only node having no
predecessor. The relation $\dashrightarrow$ is an extra relation
between nodes that will become clear later.

For such a tree to represent a strategy, we~require that each internal
node~$\mathbf{n}$ that is labelled by a $\pO$-configuration~$(s,c)$
has only one successor~$\mathbf{n'}$, with $\lambda(\mathbf
{n'})=(s',c')$ such that there is a transition $((s,c),c',(s',c'))$ in
the game~$G'$; we~require that each internal node~$\mathbf{n}$ that is
labelled with a $\pI$-state~$(s,c)$ has exactly one successor per
transition $((s,c),c',(s',c'))$ in~$G'$, each successor being labelled
with its associated~$(s',c')$. Finally, we require that each leaf~$\mathbf{n}$ of~$\calT$ has a
(strict) ancestor node~$\mathbf{n'}$ such that $\lambda(\mathbf{n'}) =
\lambda(\mathbf{n})$. The relation $\dashrightarrow$ will serve
witnessing that property. So we assume that for every leaf $\mathbf
n$, there is a unique ancestor node $\mathbf n'$ such that $\mathbf n
\dashrightarrow \mathbf n'$; furthermore it should be the case that
$\lambda(\mathbf{n'}) = \lambda(\mathbf{n})$. Under all these constraints, $\calT$ is called a \emph{strategy
  tree}. It basically represents a (finite) memory structure for a
strategy, as we now explain.

Let $\calT = (\calN,\calE,\lambda,\dashrightarrow)$ be
strategy tree for $G'$. We define $\calG_\calT = (\calN,\calE')$, a directed
graph obtained from~$\calT$ by adding
extra edges $(\mathbf n, \mathbf{n''})$ for each leaf $\mathbf n$ and
node~$\mathbf{n''}$ for which there exists another node~$\mathbf{n'}$
satisfying $\mathbf n\dashrightarrow \mathbf{n'}$ and
${(\mathbf{n'},\mathbf{n''})\in\calE}$. We~refer to these extra edges as \emph{back-edges}.  One may notice
that for any $(\mathbf n,\mathbf {n'})\in \calE'$ there is an edge from
$\lambda(\mathbf n)$ to $\lambda(\mathbf {n'})$ in~$G'$. Given two nodes $\mathbf n$ and $\mathbf {n'}$ such that $\mathbf {n'}$ is
an ancestor of~$\mathbf n$ in~$\calT$, we~write $[\mathbf {n'} \leadsto
\mathbf n]$ for the play prefix from~$\mathbf {n'}$ to~$\mathbf n$
(inclusive) using only transitions from~$\calE$.

Now, we associate with any~prefix~$\rho$ in~$\calG_\calT$
from~$\nroot$ a prefix~$\overline\rho$ in~$G'$
from~$\lambda(\nroot)=(\sroot, \croot)$ such that
$\last(\overline\rho) = \lambda(\last(\rho))$. The construction is inductive:
\begin{itemize}
\item with the empty prefix in~$\calG_\calT$ we~associate the empty
  prefix in~$G'$: $\overline{\epsilon_{\nroot}} = \epsilon_{(\sroot, \croot)}$,
\item if $\rho=\rho'\cdot (\mathbf {n'}, \mathbf n)$ with $(\mathbf
  n,\mathbf {n'})\in\calE'$, writing $(s',c')=\lambda(\mathbf{n'})$ and
  $(s,c)=\lambda(\mathbf n)$, then $\overline\rho =
  \overline{\rho'}\cdot ((s',c'),c,(s,c))$ (which by construction is
  indeed a prefix in~$G'$).
\end{itemize}

We now explain how $\calG_\calT$ corresponds to a strategy in~$G'$:
for any prefix~$\rho$ in~$\calG_\calT$, if $\lambda(\last(\rho))=(s,c)
\in \Gamma_0$, then $\last(\rho)$~has a unique successor~$\mathbf {n'}$
in~$\calG_\calT$, and, writing $(s',c')=\lambda(\mathbf {n'})$,
we~define $\sigma_\calT(\overline\rho)=((s,c),c',(s',c'))$:
$\sigma_\calT$ is a (partially defined) strategy in $G'$. The following lemma states that $\calG_\calT$ actually represents the
outcomes of the well-defined strategy~$\sigma_\calT$ from
$\lambda(\nroot)$ in $G'$: 
\begin{restatable}{lemma}{lemmastrattree}
  \label{lemma-strattree}
  Let $\mu$ be a prefix in~$G'$ from~$(\sroot,\croot)$.  Assume that
  for any $i\leq \length(\mu)$ such that $\last(\mu_{\leq i})\in
  \Gamma_0$, the~function~$\sigma_\calT$ is defined on~$\mu_{\leq i}$
  and $\mu_{\leq i+1}=\mu_{\leq i}\cdot \sigma_{\calT}(\mu_{\leq i})$.
  Then there exists a unique prefix~$\rho$ in~$\calG_\calT$ such that
  $\mu=\overline\rho$. Moreover, if $\last(\mu)\in \Gamma_0$, then $\sigma_\calT(\mu)$ is
  defined.
\end{restatable}

\begin{proof}
  The proof is by induction on~$\length(\mu)$. The base case where
  $\mu$~is the empty prefix is strightforward.

  Now, assume that the result holds for a prefix~$\mu'$, and consider a
  prefix~$\mu=\mu'\cdot((s',c'),c,(s,c))$ satisfying the hypotheses of
  the lemma. Then~$\mu'$ also satisfies those hypotheses, and we have
  $\mu'=\overline{\rho'}$ for some prefix~$\rho'$ in~$\calG_\calT$. Then, each successor~$\mathbf n$ of~$\mathbf
  n'=\last(\rho')$ is associated with a transition of~$G'$ from
  $(s',c')=\lambda(\mathbf n')$. We~consider two cases:
  \begin{itemize}
  \item if $(s',c')\in\Gamma_0$ is a $\pO$-state, then, by
    hypothesis, $\sigma_\calT(\mu')=((s',c'),c,(s,c))$, and the
    unique successor node~$\mathbf n$ of~$\mathbf n'$ in~$\calG_\calT$ is
    indeed labelled with~$(s,c)$. We~can thus let $\rho=\rho'\cdot
    (\mathbf n',\mathbf n)$.
  \item if $(s',c')\notin\Gamma_0$ belongs to~$\pI$, then each
    transition of~$G'$ corresponds to a successor~$\mathbf n$
    of~$\mathbf n'$. In~particular, there is a unique successor~$\mathbf n$
    corresponding to $((s',c'),c,(s,c))$, and again we can let
    $\rho=\rho'\cdot (\mathbf n',\mathbf n)$. 
  \end{itemize}
The last statement follows by noticing that $\sigma_\calT(\overline\rho)$ is
defined for all prefixes~$\rho$ such that
$\lambda(\last(\rho))\in\Gamma_0$.\qed
\end{proof}

We now give conditions for $\sigma_\calT$ to be a winning
strategy from~$(\sroot,\croot)$ in $G'$. With a finite
outcome~$\mu=\overline\rho$ of~$\sigma_\calT$ from~$(\sroot,\croot)$, we~associate a sequence $\decomp_\calT(\mu)$ of cycles in~$G'$,
defined inductively as follows:
\begin{itemize}
\item $\decomp_\calT(\epsilon_{(\sroot, \croot)})$ is empty;
\item if $\rho=\rho'\cdot(\mathbf{n'},\mathbf n)$ and $\mathbf n$ is
  not a leaf of~$\calT$, then
  $\decomp_{\calT}(\overline\rho)=\decomp_{\calT}(\overline{\rho'})$;
\item if $\rho=\rho'\cdot(\mathbf{n'},\mathbf n)$ and $\mathbf n$ is a
  leaf of~$\calT$, we~let $\mathbf{n''}$ be such that $\mathbf n
  \dashrightarrow \mathbf{n''}$; the~prefix $[\mathbf
   {n''}\leadsto\mathbf n]$ in~$\calT$ corresponds to a cycle~$C$
  in~$G'$, and we let $\decomp_\calT(\rho)=\decomp_\calT(\rho')\cdot
  C$.
\end{itemize}
The sequence $\decomp_\calT(\rho)$ hence contains all full cycles
(closed at leaves) encountered while reading $\rho$ in $\calT$: hence
it comprises all edges of $\rho$ except a prefix starting at $\nroot$
and a suffix since the last back-edge has been taken. It is not hard
to see that those can actually be concatenated. By~induction, we can easily show:
\begin{restatable}{proposition}{propdecomp}
  \label{prop:resume-decomp}
  Let $\mu$ be a non-empty finite outcome of $\sigma_{\calT}$ from
  $(\sroot,\croot)$ in $G'$. Write $\decomp_{\calT}(\mu) = C_0
  \cdot C_1 \cdot \ldots \cdot C_h$ (where each $C_i$
  is a cycle).  Let~$\rho$ be the prefix in~$\calG_\calT$ such that
  $\mu=\overline\rho$, $\mathbf n=\last(\rho)$, and $\nu=[\nroot
  \leadsto \mathbf n]$.  Write
  $(s_j,c_j)=\lambda(\hat\nu_j)$. 
  Then:
  \[
  \MP(\mu) = \frac{\sum_{i=0}^h \MP(C_i) \cdot \length(C_i)
    + \sum_{j=1}^{\length(\nu)} c_j 
  }{\length(\mu)}
  \]
\end{restatable}

We~say that a tree is \emph{good} if, for
every $\mathbf{n} \dashrightarrow \mathbf{n'}$ in~$\calT$,
writing $\rho=[\mathbf{n'} \leadsto \mathbf n]$ and letting
$\lambda(\hat\rho_j)=(s_j,c_j)$, it holds 
\(
\frac{\sum_{j=1}^{\length(\rho)} c_j}{\length(\rho)} \le t
\). 

\begin{restatable}{proposition}{propgoodtreestratwinning}
  \label{prop:goodtree-stratwinning}
  If $\calT$ is a finite good strategy tree, then $\sigma_{\calT}$ is a winning
  strategy from $(s_{\mathsf{root}},c_{\mathsf{root}})$ in $G'$.
\end{restatable}

\begin{proof}
  We apply Proposition~\ref{prop:resume-decomp}; the sum
  $\sum_{j=1}^{\length(\nu)} c_j$ is bounded because $\calT$~is
  finite, and $\MP(C_i)\cdot \length(C_i)\leq t\cdot
  \length(C_i)$ since $\calT$ is good. Hence for any infinite
  outcome~$\mu$ of~$\sigma_\calT$ from~$(\sroot,\croot)$, we~get
  $\MPsup(\mu)\leq t$. \qed
\end{proof}

Note that $\calT$ can be interpreted as a \emph{finite memory
  structure} for strategy $\sigma_\calT$: to know which move is given
by $\sigma_\calT$, it is sufficient to move a token in tree $\calT$,
going down in the tree, and following back-edges when stuck at leaves.

\subsection{Analyzing Winning Strategies}
\label{ssec-analysis}

We proved in the previous section that the finite-memory strategy
associated with a finite good strategy tree is winning. In this section, we first show the converse direction, proving that
from a winning strategy, we~can obtain a finite good tree. In~that~tree, backward edges correspond to (short) good
  cycles. We~will then use the result of Section~\ref{sec-bounding}
  for showing that those parts of the tree that do not belong to a segment
  $[\mathbf n \leadsto \mathbf{n'}]$ with $\mathbf{n'} \dashrightarrow
  \mathbf n$ can be replaced with other substrategies in which the
  counter value is uniformly bounded. That~way, we show that if there
  is a winning strategy for our games, then there is one where the
  counter value is uniformly bounded along all outcomes. This will allow to apply algorithms for solving games with
average-energy payoff under lower- and upper-bound
constraints~\cite{BMRLL16}.

Fix a winning strategy~$\sigma$ for~$\pO$ from~$(s_0,c_0)$
in $G'$. We~let
\begin{xalignat*}1
\finite(\sigma) = \{\pi_{\leq n} \mid\; & \pi \in
\outs((s_0,c_0),\sigma)
\text{, there is}\ m<n\ \text{s.t.} \pi_{[m+1,n]}\ \text{is a good cycle with}\\ 
& \text{no good strict sub-cycle and } \pi_{\leq n-1} \text{ does not contain any good cycle}
\}
\end{xalignat*}
and for every $\pi_{\leq n} \in \finite(\sigma)$, we define
$\mathit{back}(\pi_{\leq n})$ as $\pi_{\leq m}$ such that $\pi_{[m+1,n]}$
is a good cycle with no good strict sub-cycle.

We build a strategy tree $\calT_{\sigma}$ as follows:
\begin{itemize}
\item the nodes of~$\calT_\sigma$ are all the prefixes of the finite
  plays in~$\finite(\sigma)$; the edges relate each prefix of
  length~$k$ to its extensions of length~$k+1$. For a node~$\mathbf n$ corresponding to
  prefix~$\rho_{\leq k}$ (with~$\rho\in\finite(\sigma)$), we~let
  $\lambda(\mathbf n)= \last(\rho_{\leq k})$; we~let
  $\lambda(\nroot)=(s_0,c_0)$. The~leaves
  of~$\calT_\sigma$ then correspond to the play prefixes that are
  in~$\finite(\sigma)$.
\item backward edges in~$\calT_\sigma$ are defined by noticing that each
  leaf~$\mathbf n$ of~$\calT_\sigma$ corresponds to a finite path~$\pi_{\leq
    n}$ in~$\finite(\sigma)$, so that the prefix~$\pi_{\leq
    m}=\mathit{back}(\pi_{\leq n})$ is associated with some
  node~$\mathbf m$ of~$\calT_\sigma$ such that $\pi_{[m+1,n]}$ is a good
  cycle. This implies $\lambda(\pi_{\leq n})=\lambda(\pi_{\leq m})$,
  and we define $\mathbf n \dashrightarrow \mathbf m$. This way,
  $\calT_\sigma$ is a strategy tree as defined in
  Section~\ref{ssec-finitetrees}. 
\end{itemize}

\begin{restatable}{lemma}{lemmaTsigma}
  \label{lemma:Tsigma}
  Tree $\calT_\sigma$ is a finite good strategy tree.
\end{restatable}

\begin{proof}
  If $\calT_\sigma$ were infinite, by K\H{o}nig's Lemma it would contain
  an infinite branch. This infinite branch would correspond to an
  infinite outcome of~$\sigma$ containing no good cycle, contradicting
  Lemma~\ref{lemma:good_cycle}. 

  The fact that the tree is good is because all prefixes selected by
  $\finite(\sigma)$ correspond to good cycles.
  \qed
\end{proof}

Applying Proposition~\ref{prop:goodtree-stratwinning}, we immediately get:
\begin{corollary}\label{coro-finite-mem}
  Strategy $\sigma_{\calT_{\sigma}}$ is a winning strategy from
  $(s_0,c_0)$.
\end{corollary}

\begin{figure}[t]
  \centering
  \scalebox{1}{\begin{tikzpicture}[yscale=.5,xscale=.8]
    \path[use as bounding box] (-3,-.3) -- (9,-7.3);
    \draw (0,0) node[minirond,bleu] (root) {} node[right=3mm] {$\nroot$};
    \draw (-2,-1) node[minirond,gris] (s11) {} node[coordinate] (n11) {};
    \draw (-1,-2) node[minirond,rouge] (s22) {} node[coordinate] (n22) {};
    \draw (-2.5,-3) node[minirond,rouge] (s32) {} node[coordinate] (n32) {};
    \draw (-.5,-4) node[minirond,jaune] (s44) {} node[coordinate] (n44) {};
    \draw (-2.5,-4.5) node[minirond,gris] (s51) {} node[coordinate] (n51) {};
    \draw (-2.5,-5) node[minirond,bleu] (s52) {};
    \draw (-3.2,-6) node[minirond,rouge] (s62) {} node[coordinate] (n62) {};
    \draw (-3.2,-7) node[minirond,jaune] (s72) {} node[coordinate] (n72) {};
    \draw (-1.2,-6.5) node[minirond,jaune] (s73) {} node[coordinate] (n73) {};
    
    \begin{scope}[yshift=5mm]
    \draw (2.5,-2) node[minirond,gris] (s24) {};
    \draw (1.5,-3) node[minirond,gris] (s36) {};
    \draw (3.5,-3) node[minirond,rouge] (s38) {} node[coordinate] (n38) {};
    \draw (1,-4) node[minirond,rouge] (s45) {} node[coordinate] (n45) {};
    \draw (2.2,-4.5) node[minirond,rouge] (s46) {} node[coordinate] (n46) {};
    \draw (3.5,-4) node[minirond,gris] (s47) {} node[coordinate] (n47) {};
    \draw (3.2,-4.5) node[minirond,bleu] (s48) {};
    \draw (4,-5) node[minirond,gris] (s57) {} node[coordinate] (n57) {};
    \draw (4,-6) node[minirond,jaune] (s67) {} node[coordinate] (n67) {};
    \draw (1,-5.5) node[minirond,jaune] (s65) {} node[coordinate] (n65) {};
    \draw (2.2,-5.5) node[minirond,jaune] (s66) {} node[coordinate] (n66) {};
    \draw (3.2,-5.3) node[minirond,rouge] (s55) {} node[coordinate] (n55) {};
    \draw (3.2,-6) node[minirond,rouge] (s75) {} node[coordinate] (n75) {};
    \draw (3.2,-6.5) node[minirond,gris] (s85) {} node[coordinate] (n85) {};
    \draw (2.7,-7.7) node[minirond,jaune] (s90) {} node[coordinate] (n90) {};
    \draw (3.6,-7.3) node[minirond,jaune] (s91) {} node[coordinate] (n91) {};
    \end{scope}
    
    \draw (root) -- (s11);
    \draw (root) -- (s24);
    \draw (s11) -- (s22);
    \draw (s11) -- (s32);
    \draw (s24) -- (s36);
    \draw (s24) -- (s38);
    \draw (s22) -- (s44); 
    \draw (s32) -- (s51);
    \draw (s51) -- (s52);
    \draw (s52) -- (s62);
    \draw (s51) -- (s73); %
    \draw (s62) -- (s72); %
    
    \draw (s36) -- (s45);
    \draw (s36) -- (s46);
    \draw (s38) -- (s47);
    \draw (s47) -- (s48);
    \draw (s47) -- (s57);
    \draw (s57) -- (s67); 
    \draw (s45) -- (s65); 
    \draw (s46) -- (s66); 
    \draw (s48) -- (s55);
    \draw (s55) -- (s75);
    \draw (s75) -- (s85);
    \draw (s85) -- (s90);
    \draw (s85) -- (s91);

    \draw (s72) edge[dashed,-latex',draw=red!50!black,out=140,in=-140] (s62);
    \draw (s73) edge[dashed,-latex',draw=red!50!black,out=60,in=-50] (s32);
    \draw (s44) edge[dashed,-latex',draw=red!50!black,out=60,in=-50] (s22);
    \draw (s65) edge[dashed,-latex',draw=red!50!black,out=140,in=-140] (s45);
    \draw (s66) edge[dashed,-latex',draw=red!50!black,out=50,in=-50] (s46);
    \draw (s90) edge[dashed,-latex',draw=red!50!black,out=140,in=-140] (s75);
    \draw (s91) edge[dashed,-latex',draw=red!50!black,out=50,in=-50] (s55);
    \draw (s67) edge[dashed,-latex',draw=red!50!black,out=50,in=-50] (s38);

    \begin{scope}[line width=2.4mm,green!70!black!20!white,round cap-round cap,shorten >=-2mm,shorten <=-2mm]
    \draw (n32) -- (n51) -- (n73);
    \draw (n62) -- (n72);
    \draw (n22) -- (n44);
    \draw (n65) -- (n45);
    \draw (n66) -- (n46);
    \draw (n75) -- (n85) -- (n90);
    \draw (n55) -- (n85) -- (n91);
    \draw (n38) -- (n47) -- (n57) -- (n67);
    \end{scope}
    \draw (-2,-1) node[minirond,gris] (s11) {} node[coordinate] (n11) {};
    \draw (-1,-2) node[minirond,rouge] (s22) {} node[coordinate] (n22) {};
    \draw (-2.5,-3) node[minirond,rouge] (s32) {} node[coordinate] (n32) {};
    \draw (-.5,-4) node[minirond,jaune] (s44) {} node[coordinate] (n44) {};
    \draw (-2.5,-4.5) node[minirond,gris] (s51) {} node[coordinate] (n51) {};
    \draw (-2.5,-5) node[minirond,bleu] (s52) {};
    \draw (-3.2,-6) node[minirond,rouge] (s62) {} node[coordinate] (n62) {};
    \draw (-3.2,-7) node[minirond,jaune] (s72) {} node[coordinate] (n72) {};
    \draw (-1.2,-6.5) node[minirond,jaune] (s73) {} node[coordinate] (n73) {};
    
    \begin{scope}[yshift=5mm]
    \draw (2.5,-2) node[minirond,gris] (s24) {};
    \draw (1.5,-3) node[minirond,gris] (s36) {};
    \draw (3.5,-3) node[minirond,rouge] (s38) {} node[coordinate] (n38) {};
    \draw (1,-4) node[minirond,rouge] (s45) {} node[coordinate] (n45) {};
    \draw (2.2,-4.5) node[minirond,rouge] (s46) {} node[coordinate] (n46) {};
    \draw (3.5,-4) node[minirond,gris] (s47) {} node[coordinate] (n47) {};
    \draw (3.2,-4.5) node[minirond,bleu] (s48) {};
    \draw (4,-5) node[minirond,gris] (s57) {} node[coordinate] (n57) {};
    \draw (4,-6) node[minirond,jaune] (s67) {} node[coordinate] (n67) {};
    \draw (1,-5.5) node[minirond,jaune] (s65) {} node[coordinate] (n65) {};
    \draw (2.2,-5.5) node[minirond,jaune] (s66) {} node[coordinate] (n66) {};
    \draw (3.2,-5.3) node[minirond,rouge] (s55) {} node[coordinate] (n55) {};
    \draw (3.2,-6) node[minirond,rouge] (s75) {} node[coordinate] (n75) {};
    \draw (3.2,-6.5) node[minirond,gris] (s85) {} node[coordinate] (n85) {};
    \draw (2.7,-7.7) node[minirond,jaune] (s90) {} node[coordinate] (n90) {};
    \draw (3.6,-7.3) node[minirond,jaune] (s91) {} node[coordinate] (n91) {};
    \end{scope}
    \draw (root) -- (s11);
    \draw (root) -- (s24);
    \draw (s11) -- (s22);
    \draw (s11) -- (s32);
    \draw (s24) -- (s36);
    \draw (s24) -- (s38);
    \draw (s22) -- (s44); 
    \draw (s32) -- (s51);
    \draw (s51) -- (s52);
    \draw (s52) -- (s62);
    \draw (s51) -- (s73); %
    \draw (s62) -- (s72); %
    
    \draw (s36) -- (s45);
    \draw (s36) -- (s46);
    \draw (s38) -- (s47);
    \draw (s47) -- (s48);
    \draw (s47) -- (s57);
    \draw (s57) -- (s67); 
    \draw (s45) -- (s65); 
    \draw (s46) -- (s66); 
    \draw (s48) -- (s55);
    \draw (s55) -- (s75);
    \draw (s75) -- (s85);
    \draw (s85) -- (s90);
    \draw (s85) -- (s91);

    \draw (s72) edge[dashed,-latex',draw=red!50!black,out=140,in=-140] (s62);
    \draw (s73) edge[dashed,-latex',draw=red!50!black,out=60,in=-50] (s32);
    \draw (s44) edge[dashed,-latex',draw=red!50!black,out=60,in=-50] (s22);
    \draw (s65) edge[dashed,-latex',draw=red!50!black,out=140,in=-140] (s45);
    \draw (s66) edge[dashed,-latex',draw=red!50!black,out=50,in=-50] (s46);
    \draw (s90) edge[dashed,-latex',draw=red!50!black,out=140,in=-140] (s75);
    \draw (s91) edge[dashed,-latex',draw=red!50!black,out=50,in=-50] (s55);
    \draw (s67) edge[dashed,-latex',draw=red!50!black,out=50,in=-50] (s38);

    \begin{scope}[xshift=5cm,yscale=.8]
      \draw (0,0) node[minirond,jaune] {}
        node[right=2mm, text width=2cm] {leaf};
      \draw (0,-1) node[minirond,rouge] {}
        node[right=2mm, text width=3cm] {start of good cycle};
      \draw (0,-2) node[minirond,bleu] {}
        node[right=2mm, text width=2cm] {critical node};
      \draw[dashed,-latex',draw=red!50!black] (-.3,-2.9) -- +(.6,0)
        node[right, text width=3cm] {backward edge};
      \draw[line width=1.4mm,green!70!black!20!white] (-.3,-3.9) -- +(.6,0)
        node[right, text width=2cm,black] {good cycle};
    \end{scope}
  \end{tikzpicture}}
  \caption{Example of a finite strategy tree, with backward edges and
    critical nodes.}\label{fig-critnodes}
\end{figure}
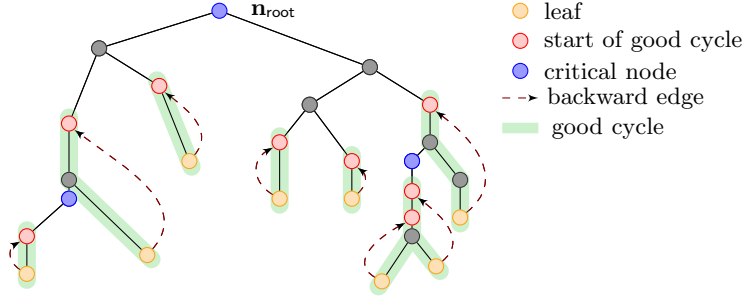

Let $\mathbf{n} \dashrightarrow \mathbf{n'}$ be two related nodes
in~$\calT_\sigma$. We say that a node~$\mathbf{n''}$ is \emph{just
  below} $[\mathbf{n'} \leadsto \mathbf{n}]$ in~$\calT_\sigma$
whenever its predecessor appears along $[\mathbf{n'} \leadsto
  \mathbf{n}]$, but node~$\mathbf{n''}$ itself does not appear along
any path $[\mathbf{n_1} \leadsto \mathbf{n_2}]$ for which
$\mathbf{n_2} \dashrightarrow \mathbf{n_1}$. Such nodes, together with
the root of the tree, are called the \emph{critical}
nodes. Fig.~\ref{fig-critnodes} illustrates this definition.

\begin{restatable}{lemma}{lemmaboundcrit}
  If $\mathbf{n}$ is a critical node in~$\calT_\sigma$, then writing
  $\lambda(\mathbf{n}) = (s,c)$, we have that $c \le t + W \cdot
  (N+1)$.
\end{restatable}

\begin{proof}
By~Proposition~\ref{coro-N}, for every $\mathbf{n} \dashrightarrow
\mathbf{n'}$, the length of the prefix $[\mathbf{n'}\leadsto
  \mathbf{n}]$ in~$\calT_\sigma$ is at most~$N$. So the length of the prefix from~$\mathbf n'$ to a critical node is at most~$N+1$. Since node~$\mathbf n'$ corresponds to the beginning of a good cycle, we have $\lambda(\mathbf n')=(s',c')$ with $c'\leq t$. Finally, since edges in~$G'$ have weight at most~$W$, we~get our result.\qed
\end{proof}

Given a critical node~$\mathbf{n}$, we define
\begin{multline*}
  \target(\mathbf{n}) = \{\mathbf{n'} \text{ in subtree of }\mathbf{n}\mid
   \text{there exists $\mathbf{n''}$ such that }\mathbf{n''}\dashrightarrow \mathbf{n'} \\
  \text{and}\ [\mathbf{n} \leadsto \mathbf{n'}]\ \text{contains no other such
    node}\}.
\end{multline*}
Looking again at Fig.~\ref{fig-critnodes}, the targets of a critical
node are the start nodes of the good cycles that are closest to that
critical node.
In~particular, for the rightmost critical node on
Fig.~\ref{fig-critnodes}, there are two candidate target nodes
(because there are two overlapping good cycles), but only the topmost
one is a target.

For every critical node~$\mathbf{n}$, we apply
Lemma~\ref{lemma:pushdown_games} with $\gamma = \lambda(\mathbf{n})$
and $\Gamma'_\gamma = \{\gamma' = \lambda(\mathbf{n'}) \mid
\mathbf{n'} \in \target(\mathbf{n})\}$, setting $M = t + W \cdot
(N+1)$.  We write $\sigma_{\mathbf{n}}$ for the corresponding
strategy: applying $\sigma_{\mathbf{n}}$ from $\lambda(\mathbf{n})$,
player $P_0$ will reach some configuration $(s',c')$ such that there
is a node $\mathbf{n'} \in \target(\mathbf{n})$ with
$\lambda(\mathbf{n'}) = (s',c')$.

\looseness=-1
Now, for any node $\mathbf{n'}$ that is the target of a backward edge
$\mathbf{n} \dashrightarrow \mathbf{n'}$, but whose immediate
predecessor does not belong to any segment $[\mathbf{n_1} \leadsto
  \mathbf{n_2}]$ with $\mathbf{n_2} \dashrightarrow \mathbf{n_1}$, we
define strategy $\sigma_{[\mathbf{n'}]}$ which follows good cycles as
much as possible; when a leaf~$\mathbf m$ is reached, the strategy
replays similarly as from the equivalent node $\mathbf{m'}$ for which
$\mathbf m \dashrightarrow \mathbf{m'}$.
If,~while playing that strategy, the play ever leaves a good cycle 
(due to a move of player~$P_1$), then it reaches
a critical node $\mathbf{n''}$. From that node, we will apply strategy
$\sigma_{\mathbf{n''}}$ as defined above, and iterate like this.

This defines a strategy~$\sigma'$. Applying
Lemma~\ref{lemma:pushdown_games} to strategies
$\sigma_{\mathbf{n}}$ when $\mathbf{n}$ is critical, and the previous
analysis of good cycles, we get the following doubly-exponential bound
on the counter value (which is only exponential in case constants $W$, $t_1$, and~$t_2$ are encoded in unary):

\begin{restatable}{proposition}{propboundM}
\label{prop:boundM}
  Strategy $\sigma'$ is a winning strategy from $(s_0,c_0)$, and all
  visited configurations $(s,c)$ when applying $\sigma'$ are such that
  $c \le M'$ with 
  \[
  M' = 2^{\mathcal{O}(t+W \cdot (8 t_1 t_2 (t+1)^3 |S|^2+1) + \size{S}
    + \size{E} \cdot W + \size{S} \cdot (\lceil t\rceil +1))}.
  \]
\end{restatable}

\subsection{Conclusion}

Gathering everything we have done above, we get the following equivalence.

\begin{proposition}
  Player $\pO$ has a winning strategy in game $G$ from $\initState$
  for the objective $\AvgLower(t)$ if, and only if, he has a wining
  strategy in $G$ from $\initState$ for the objective
  $\AvgLowerUpper(t,U) =\LUBound(U) \cap
  \AvgEnergyLevel(t)$, where $U = M'$ is the bound from Proposion~\ref{prop:boundM}.
\end{proposition}

Hence we can use the algorithm for games with objectives
$\AvgLowerUpper(t,U)$ in~\cite{BMRLL16}, which is polynomial  in $\size{S}$, $\size{E}$, $t$, and $U$ (hence pseudo-polynomial only). Having in mind that the upper bound $U$ is doubly-exponential, we can deduce our main
decidability result. The memory required is also a consequence
of~\cite{BMRLL16}.

\begin{theorem}
  The AEL threshold problem is in \EXPTIME[2].  Furthermore
  dou\-bly-exponential memory is sufficient to win (for player $\pO$).
\end{theorem}

In~\cite{Hunter14arxiv}, a super-exponential lower bound is given for
the required memory to win a succinct one-counter game. While the
model of games is not exactly the same, the actual family of games
witnessing that lower bound on the memory happens to be usable as well
for the AEL threshold problem (with threshold zero). The reduction is similar to the one in the proof of Theorem~\ref{thm:hardness}. This yields a
lower bound on the required memory to win games with
$\AvgEnergyLevel(t)$ objectives which is $2^{({2^{\sqrt{n}}}/{\sqrt{n}})-1}$.

For unary encodings or small weights we get better results from our technique:

\begin{corollary}
The AEL threshold problem is in \EXPTIME and
  exponential memory is sufficient to win (for player $\pO$), if the weights and the threshold are encoded in unary or polynomial in the size of the graph.
\end{corollary}

\section{EXPSPACE-hardness}
\label{sec-hard}
In this section, we show that the AEL threshold problem is $\EXPSPACE$-hard by a reduction from succinct one-counter reachability games with counter values in $\mathbb{N}$, thereby improving the previously best $\EXPTIME{}$ lower bound. Such a game is played in a graph~$(S_0, S_1, E)$ as defined above, but without parallel edges. However, the semantics are slightly different: assume a play prefix~$\rho$ has been produced by the players thus far. Then, the player whose turn it is may only pick an edge~$e$ such that $\EL(\rho) + \weg(e) \ge 0$, i.e., edges that would lead to a negative energy level are disabled. If the player has no enabled edges available, then he loses immediately. Note that this indeed differs from the semantics of the games we consider here, where a negative energy level is a direct loss for $\pO$, no matter who picked the edge leading to the negative level. 
The goal of $\pO$ in a succinct one-counter game for is to reach, from a given initial state~$\initState \in S$ with energy level zero, a state~$s \neq \initState$ with energy level zero.

\begin{proposition}[\cite{Hun15}]
The following problem is $\EXPSPACE{}$-complete: given a graph~$(S_0, S_1, E)$ without parallel edges and a vertex~$\initState \in S$, has $\pO$ a strategy (respecting the blocking semantics described above) so that every outcome is infinite and has a prefix that ends in a state~$s \neq \initState$ and has energy level zero.
\end{proposition}

We reduce this problem to the AEL threshold problem by simulating a one-counter reachability game by an AEL game: the reachability objective can easily be encoded by giving $\pO$ the ability to stop the simulation at any state~$s\neq \initState$ by going (with weight zero) to a fresh sink state~$s_{\sf sink}$ that is equipped with a self-loop of weight zero. Then, if he ever encounters a prefix of energy level zero, he can stop the simulation, go to the sink state, and achieve an average-energy of zero. On the other hand, if he never reaches energy level zero during the simulation, he has two options: either simulating the one-counter game ad infinitum while never reaching energy level zero or going to the sink state with non-zero energy level. In both cases, the average-energy is greater than zero. Thus, the reachability objective can be taken care of by an average-energy objective with threshold zero. 

It remains to show how to implement the disabling of edges leading to a negative energy level. As already remarked above, such edges lead to a direct loss for $\pO$ in an average-energy game. Thus, a winning strategy for $\pO$ never picks a disabled edge, i.e., we only have to consider the case of disabled edges at vertices of $\pI$. Intuitively, we do the following: every time $\pI$ picks an edge~$e$ with negative weight~$\weg(e)$, $\pO$ has the choice to let this edge be taken or to claim that the energy level is strictly below~$-\weg(e)$. If his claim is correct, he will be able to reach the sink state~$s_{\sf sink}$ with energy level zero, and thereby ensure an average-energy of zero. In contrast, if his claim is incorrect, then he either reaches a negative energy level or cannot ensure an average-energy of zero. 
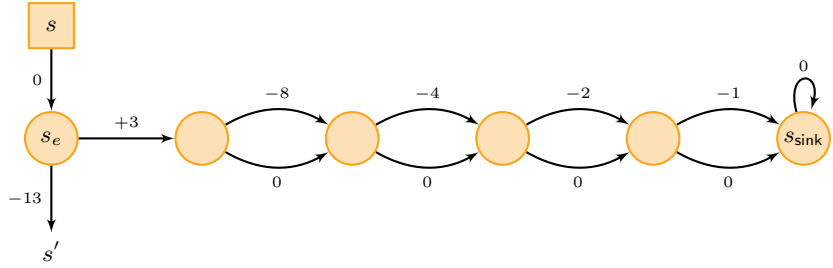
\begin{figure}[ht]
\centering

\begin{tikzpicture}[auto,node distance=2 cm, thick,every loop/.append style={-latex'}]

\node[carre,jaune] (v) at (0,1.5) {$s$};
\node[rond,jaune] (se) at (0,0) {$s_e$};
\node[]		(v') at (0,-1.5) {$s'$};	

\node[rond,jaune] (8) at (2,0) {};
\node[rond,jaune] (4) at (4,0) {};
\node[rond,jaune] (2) at (6,0) {};
\node[rond,jaune] (1) at (8,0) {};
\node[rond,jaune] (s) at (10,0) {$s_{\sf sink}$};

\path[-latex']
(v)  edge node[left] {$\scriptstyle 0$} (se)
(se) edge node[left] {$\scriptstyle -13$} (v')
(se) edge node[above] {$\scriptstyle  +3$} (8)
(8) edge[bend left] node[above] {$\scriptstyle -8$} (4)
(8) edge[bend right] node[below] {$\scriptstyle 0$} (4)
(4) edge[bend left] node[above] {$\scriptstyle -4$} (2)
(4) edge[bend right] node[below] {$\scriptstyle 0$} (2)
(2) edge[bend left] node[above] {$\scriptstyle -2$} (1)
(2) edge[bend right] node[below] {$\scriptstyle 0$} (1)
(1) edge[bend left] node[above] {$\scriptstyle -1$} (s)
(1) edge[bend right] node[below] {$\scriptstyle 0$} (s)
(s) edge[loop above] node[above] {$\scriptstyle 0$}();
	\end{tikzpicture} 
      \caption{The gadget implementing the blocking of edges (here, $e = (s, -13,s')$): if a play prefix ends in $s$ with energy level $c$, then $\pO$ can reach $s_{\sf sink}$ with energy level zero while maintaining non-negative energy levels if, and only if, $c \le 12$. If $c \ge 13$ he has to go to $s'$ when maintaining non-negative energy levels. }
\label{fig-onecountergadget}

\end{figure}

As an example, assume there is an edge $e = (s, -13, s')$ in the one-counter game. It is replaced by the gadget shown in Fig.~\ref{fig-onecountergadget}, where all states but $s$ and $s'$ are fresh. It is straightforward to show that  $\pO$ can, when at from $s_e$, enter the gadget (i.e., not go to $s'$) and ensure an non-negative energy level and average energy level zero if, and only if, the current energy level at $s$ is at most~$12$, by reaching the sink state with energy level zero. On the other hand, if the energy level is larger than $12$, then he cannot reach the sink with energy level zero. Thus, he has to take the edge to $s'$ as intended by $\pI$.

 Obviously, this gadget can be generalized to arbitrary negative weights~$w(e)$ with $\log(w(e))$ many states, which may even be shared among gadgets. 

Using the encoding of the reachability objective by an average-energy objective and the blocking-gadget described above, it is straightforward to reduce succinct once-counter games with counters in $\mathbb{N}$ to the AEL threshold problem with threshold zero.

\begin{restatable}{theorem}{theoremhardness}
\label{thm:hardness}
The AEL threshold problem is $\EXPSPACE{}$-hard, even for the fixed threshold zero. 
\end{restatable}

\section{Multi-dimensional Average-energy Games}
\label{sec-undec}

We now turn to a more general class of games where integer weights on the edges are replaced by \textit{vectors of integer weights}, representing changes in different quantitative aspects. That is, for a game $\Game =
(S_0, S_1, \trans)$ of dimension $k \geq 1$, we now have $\trans \subseteq S \times [-W,W]^k\times S$ for $W \in \bbN$. Multi-dimensional games have recently gained interest as a powerful model to reason about interplays and trade-offs between different resources; and multi-dimensional versions of many classical objectives have been considered in the literature: e.g., mean-payoff~\cite{Chatterjee2013,VelnerC0HRR15}, energy~\cite{Chatterjee2013,DBLP:conf/icalp/JurdzinskiLS15,VelnerC0HRR15}, or total-payoff~\cite{DBLP:journals/iandc/Chatterjee0RR15}. 
 We consider the natural extensions of threshold problems in the multi-dimensional setting: we take the zero vector in $\bbN^k$ as lower bound for the energy, a vector $U \in \bbN^k$ as upper bound, a vector $t \in \bbQ^k$ as threshold for the average-energy, and the payoff functions are defined using component-wise limits. That is, we essentially take the \textit{conjunction} of our objectives for all dimensions. We quickly review the situation for the three types of average-energy objectives.

\paragraph*{Average-energy games (without energy bounds).} In the one-dimensional version of such games, memoryless strategies suffice for both players and the threshold problem is in $\NP\,\cap\,\coNP$~\cite{BMRLL16}. We prove here that already for games with three dimensions, the threshold problem is undecidable, based on a reduction from two-dimensional robot games~\cite{RobotGames}. Decidability for average-energy games with two dimensions remains open.

\begin{restatable}{theorem}{thmAEundecidable}
\label{thm:AEundecidable}
The threshold problem for average-energy games with three or more dimensions is undecidable. That is, given a finite $k$-dimensional game~$\Game = (S_0, S_1, E)$, for $k \geq 3$, an~initial state $\initState \in
\states$, and a threshold $t
\in \bbQ^k$, determining whether~$\pO$ has a winning
strategy from~$\initState$ for objective~$\AvgEnergyLevel(t)$ is undecidable.
\end{restatable}

\begin{proof}
Two-dimensional robot games~\cite{RobotGames} are a special case of counter reachability games, expressible as follows: $R = (\{q_0\}, \{q_1\}, T)$, where $q_0$ (resp.~$q_1$) is the unique state belonging to $\pO$ (resp.~$\pI$), $Q = \{q_0, q_1\}$, and $T \subseteq Q \times [-V,V]^2 \times Q$, $V \in \bbN$, is a finite set of transitions. The game starts in $q_0$ with given counter values $(x_0, y_0) \in \bbZ^2$ and, when in configuration $(q, (x, y))$, taking transition $t = (q, (a, b), q')$ takes the game to configuration $(q', (x+a, y+b))$. The goal of $\pO$ is to reach configuration $(q_0, (0, 0))$ and $\pI$ tries to prevent it. It was recently proved that deciding the winner in such games is undecidable~\cite{RobotGames}.

Given $R = (\{q_0\}, \{q_1\}, T)$ with initial configuration $(q_0, (x_0, y_0))$, we build a three-dimensional average-energy game $\Game$ with threshold $(0, 0, 0)$ such that $\pO$ wins in $\Game$ if, and only if, he wins in $R$, which implies undecidability of the AEL threshold problem. Let $\Game = (S_0, S_1, E)$ with $S_0 = \{q_{\text{init}}, q_0, q_{\text{stop}}\}$ (with $q_{\text{init}}$ and $q_{\text{stop}}$ being fresh states), $S_1 = \{q_1\}$ and $E \subseteq S \times [-W,W]^3 \times S$, $W \in \bbN$, built as follows:
\begin{itemize}
\item if $(q, (a,b), q') \in T$, then $(q, (a, b, -a -b), q') \in E$,
\item $(q_0, (-1, -1, -1), q_{\text{stop}}) \in E$ and $(q_{\text{stop}}, (0, 0, 0), q_{\text{stop}}) \in E$,
\item $(q_{\text{init}}, (x_0 + 1, y_0 + 1, - x_0 - y_0 + 1), q_0) \in E$ (where $(x_0, y_0)$ are the initial counter values in $R$).
\end{itemize}
Essentially, we add a third dimension that contains the opposite of the sum of the first two dimensions of each transition in $T$; we add the possibility to branch from $q_0$ to $q_{\text{stop}}$ using a $(-1, -1, -1)$-edge to reach an absorbing state with a $(0, 0, 0)$-loop; and we also add an initial edge that encodes the initial configuration of $R$. Now, the initial state of $\Game$ is $q_{\text{init}}$ and the threshold is $t = (0, 0, 0)$.

We proceed to prove the reduction to be correct. First, let $\St_0$ be a winning strategy for~$\AvgEnergyLevel(t)$ in $\Game$. We will show that $\pO$ can also win in $R$. We start by claiming that $\pO$ has to branch to $q_{\text{stop}}$ at some point otherwise he cannot win. If the $(-1, -1, -1)$-edge is not taken along a play $\pi$, then it holds that for all $i \geq 1$, the energy level $\EL(\pi_{\leq i}) = (x_i, y_i, z_i)$ satisfies $x_i + y_i + z_i = 3$. Indeed, it holds in the first step thanks to the initial edge and it continues to hold at each step as for each edge issued from $T$, the sum of the three dimensions is zero. Now observe the following:
  \begin{xalignat*}1
   3  &=  \limsup_{n \rightarrow \infty} \frac{1}{n} \sum_{i=1}^n (x_i + y_i + z_i)\\
   &\leq \limsup_{n \rightarrow \infty} \frac{1}{n} \sum_{i=1}^n x_i + \limsup_{n \rightarrow \infty} \frac{1}{n} \sum_{i=1}^n y_i + \limsup_{n \rightarrow \infty} \frac{1}{n} \sum_{i=1}^n z_i\\
   &= x + y + z
  \end{xalignat*}
for $(x, y, z) = \AEsup(\pi)$. Hence, at least one dimension has an average-energy strictly greater than zero (otherwise their sum cannot be greater than or equal to three), and the threshold $t = (0, 0, 0)$ cannot be met. Thus, we know that branching to $q_{\text{stop}}$ is necessary since $\St_0$ is winning. But using the decomposition lemma of~\cite{BMRLL16} and the fact that $q_{\text{stop}}$ is an absorbing state with a $(0, 0, 0)$-loop, we also know that for the average-energy to be less than $(0, 0, 0)$, the energy level when branching must be no more than $(1, 1, 1)$ (as we will use a $(-1, -1, -1)$-edge). Furthermore, if the energy level on some dimension is strictly less than one, then it must be strictly greater on another dimension (because the sum is always equal to three before branching), and we are not below $(1, 1, 1)$ either. In conclusion, strategy $\St_0$ must ensure reaching the exact configuration $(q_0, (1, 1, 1))$ in order to win. Finally, observe that by construction of our game $\Game$, reaching this configuration is equivalent to reaching $(q_0, (0, 0))$ in $R$. Hence, we can easily build a strategy $\St^R_0$ in $R$ that mimics $\St_0$ in order to win the robot game. This strategy $\St^R_0$ could in general use arbitrary memory (since we start with an arbitrary strategy $\St_0$) while robot games as defined in~\cite{RobotGames} only allow strategies to look at the current configuration. Still, from $\St^R_0$, one can easily build a corresponding strategy that meets this restriction ($R$ being a counter reachability game, there is no reason to choose different actions in two visits of the same configuration).

Similarly, from a winning strategy $\St^R_0$ in $R$, we can define a strategy $\St_0$ that mimics it in $G$ in order to reach $(q_0, (1, 1, 1))$, and at that point, branches to $q_{\text{stop}}$ and wins in $G$ for $\AvgEnergyLevel(t)$. Thus, $\pO$ wins in $G$ if, and only if, he wins in $R$, which concludes our proof.\qed
\end{proof}

\paragraph*{Average-energy games with lower and upper bounds.} 
One-dimensional versions of those games were proved to be $\EXPTIME$-complete in~\cite{BMRLL16}. The algorithm consists in reducing (in two steps) the original game to a mean-payoff game on an expanded graph of pseudo-polynomial size (polynomial in the original game but also in the upper bound $U \in \bbN$) and applying a pseudo-polynomial time algorithm for mean-payoff games (e.g.,~\cite{fmsd38(2)-BCDGR}). Intuitively, the trick is that the bounds give strong constraints on the energy levels that can be visited along a play without losing and thus one can restrict the game to a particular graph where acceptable energy levels are encoded in the states and exceeding the bounds is explicitely represented by moving to ``losing'' states, just as we did in Section~\ref{subsec:reductionToMP} for the lower bound. Carefully inspecting the construction of~\cite{BMRLL16}, we observe that the same construction can be generalized straightforwardly to the multi-dimensional setting. However, the overall complexity is higher: first, the expanded graph will be of \textit{exponential size in }$k$, the number of dimensions, while still polynomial in $S$ and $U$. Second, multi-dimensional \textit{limsup} mean-payoff games are in $\NP \cap \coNP$~\cite{VelnerC0HRR15}.

\begin{theorem}
The threshold problem for multi-dimensional average-energy ga\-mes with lower and upper bounds is in $\ComplexityFont{NEXPTIME} \cap \ComplexityFont{coNEXPTIME}$. That is, given a finite $k$-dimensional game~$\Game = (S_0, S_1, E)$, an~initial state $\initState \in
\states$, an upper bound $U \in \bbN^k$, and a threshold $t
\in \bbQ^k$, determining if~$\pO$ has a winning
strategy from~$\initState$ for objective~$\LUBound(U) \cap \AvgEnergyLevel(t)$ is in $\ComplexityFont{NEXPTIME} \cap \ComplexityFont{coNEXPTIME}$.
\end{theorem}

Whether the $\EXPTIME$-hardness that trivially follows from the one-dimen\-sion\-al case~\cite{BMRLL16} can be enhanced to meet this upper bound (or conversely) is an open problem.

\paragraph*{Average-energy games with lower bound but no upper bound.} Finally, we consider the core setting of this paper, which we just proved decidable in one-dimension, solving the open problem of \cite{BMRLL16}. Unfortunately, we show that those games are undecidable \textit{as soon as two-dimensional weights are allowed}. To prove it, we reuse some ideas of the proof of undecidability for multi-dimensional total-payoff games presented in~\cite{DBLP:journals/iandc/Chatterjee0RR15}, but specific gadgets need to be adapted. Hence we provide a full proof here.

\begin{restatable}{theorem}{thmAELundecidable}
\label{thm:AELundecidable}
The threshold problem for lower-bounded average-energy games with two or more dimensions is undecidable. That is, given a finite $k$-dimensional game~$\Game = (S_0, S_1, E)$, for $k \geq 2$, an~initial state $\initState \in
\states$, and a threshold $t
\in \bbQ^k$, determining whether~$\pO$ has a winning
strategy from~$\initState$ for objective~$\AvgLower(t)$ is undecidable.
\end{restatable}

\begin{proof}
We reduce the halting problem for two-counter machines (2CMs) to the threshold problem for lower-bounded average-energy games with two dimensions. From a two-counter machine $\calM$, we construct a two-player game $\Game$ with two dimensions and a lower-bounded average-energy objective such that $\pO$ wins for threshold $(0,0)$ if, and only if, the 2CM halts. Counters take values $(v_{1}, v_{2}) \in \bbN^{2}$ along an execution, and can be incremented or decremented (if positive). A~counter can be tested for equality to zero, and the machine can branch accordingly.
The halting problem for 2CMs is undecidable~\cite{minsky1961}. Assume w.l.o.g. that we have a 2CM $\calM$ such that if it halts, it halts with the two counters equal to zero, which is possible as it suffices to plug a machine that decreases both counters to zero at the end of the execution of the considered machine. In the game we construct, $\pO$ has to faithfully simulate the 2CM~$\calM$ and wins only if it halts. The role of $\pI$ is to ensure that he does so by retaliating if it is not the case, hence making the outcome losing for the lower-bounded average-energy objective.

The game is built as follows. The states of $\Game$ are copies of the control states of~$\calM$ plus some special states discussed in the following. Edges represent transitions between these states. Let $C_1$ and $C_2$ be the two counters, starting with value zero. We start the game by taking an edge of weight $(1, 1)$ and we build our game such that, at all times along a faithful execution, the counters have value $(v_{1}, v_{2})$ if, and only, if the current energy level is $(v_1 + 1, v_2 +1)$. Each increment of $C_1$ (resp.~$C_2$) in $\calM$ is implemented in $\Game$ by an edge of weight $(1, 0)$ (resp.~$(0, 1)$); each decrement of $C_1$ (resp.~$C_2$) by an edge of weight $(-1, 0)$ (resp.~$(0,-1)$).

We first present how to ensure a faithful simulation of the 2CM $\calM$ by $\pO$. The necessary gadgets are depicted in Fig.~\ref{fig:gadgets}.

\begin{figure}[htb]
        \centering
\subfloat[Counters are always non-negative.]{\label{fig:gadgetPos}\scalebox{1}{\begin{tikzpicture}[->,>=stealth',shorten >=1pt,auto,node
    distance=2.5cm,bend angle=45, scale=1, font=\normalsize,inner sep=.5mm, thick,every loop/.append style={-latex'}]
    \everymath{\scriptstyle}
 
    \node[carre,jaune]  (0)  at (0, 0) {};
    \node[rond,jaune]  (1) at (1.8, 0) {};
    \node[rond,jaune]  (2) at (3.6, 0) {};
    \node[rond,jaune]  (3) at (5.4, 0) {};
    
    \coordinate[shift={(0mm,5mm)}] (from) at (0.north);
    \coordinate[shift={(0mm,-5mm)}] (to) at (0.south);
    \path
    (1) edge [loop below, out=240, in=300,looseness=2, distance=8mm,-latex'] node [below] {$\scriptstyle (-1,0)$} (1)
    (2) edge [loop below, out=240, in=300,looseness=2, distance=8mm,-latex'] node [below] {$\scriptstyle (0,-1)$} (2)
    (3) edge [loop below, out=240, in=300,looseness=2, distance=8mm,-latex'] node [below] {$\scriptstyle (0,0)$} (3);
	\draw[-latex',dotted] (from) to (0);
	\draw[-latex',dotted] (0) to (to);
	\draw[-latex'] (0) to node[above,yshift=1mm] {$\scriptstyle (-1,-1)$} (1);
	\draw[-latex'] (1) to node[above,yshift=1mm] {$\scriptstyle (0,0)$} (2);
	\draw[-latex'] (2) to node[above,yshift=1mm] {$\scriptstyle (0,0)$} (3);
      \end{tikzpicture}}}
      \hspace{16mm}
      \subfloat[Halting.]{\scalebox{1}{\label{fig:gadgetHalt}\begin{tikzpicture}[->,>=stealth',shorten >=1pt,auto,node
    distance=2.5cm,bend angle=45, scale=1, font=\normalsize,inner sep=.5mm, thick,every loop/.append style={-latex'}]
    \everymath{\scriptstyle}

    \node[rond,jaune]  (0)  at (0, 0) {};
    \node[rond,jaune]  (1) at (1.6, 0) {};
    
    \coordinate[shift={(-5mm,0mm)}] (from) at (0.west);
    \path
    (1) edge [loop below, out=240, in=300,looseness=2, distance=8mm,-latex'] node [below] {$\scriptstyle (0,0)$} (1);
	\draw[-latex',dotted] (from) to (0);
	\draw[-latex'] (0) to node[above,yshift=1mm] {$\scriptstyle (-1,-1)$} (1);
      \end{tikzpicture}}}
      
      \subfloat[$C_1$ is equal to zero.]{\label{fig:gadgetZeroOne}\scalebox{1}{\begin{tikzpicture}[->,>=stealth',shorten >=1pt,auto,node
    distance=2.5cm,bend angle=45, scale=1, font=\normalsize,inner sep=.5mm, thick,every loop/.append style={-latex'}]
    \everymath{\scriptstyle}

    \node[carre,jaune]  (0)  at (0, 0) {};
    \node[rond,jaune]  (1) at (1.6, 0) {};
    \node[rond,jaune]  (2) at (3.2, 0) {};
    
    \coordinate[shift={(0mm,5mm)}] (from) at (0.north);
    \coordinate[shift={(0mm,-5mm)}] (to) at (0.south);
    \path
    (1) edge [loop below, out=240, in=300,looseness=2, distance=8mm,-latex'] node [below] {$\scriptstyle (0,-1)$} (1)
    (2) edge [loop below, out=240, in=300,looseness=2, distance=8mm,-latex'] node [below] {$\scriptstyle (0,0)$} (2);
	\draw[-latex',dotted] (from) to (0);
	\draw[-latex',dotted] (0) to (to);
	\draw[-latex'] (0) to node[above,yshift=1mm] {$\scriptstyle (-1,0)$} (1);
	\draw[-latex'] (1) to node[above,yshift=1mm] {$\scriptstyle (0,0)$} (2);
      \end{tikzpicture}}}
      \hspace{20mm}
      \subfloat[$C_2$ is equal to zero.]{\label{fig:gadgetZeroTwo}\scalebox{1}{\begin{tikzpicture}[->,>=stealth',shorten >=1pt,auto,node
    distance=2.5cm,bend angle=45, scale=1, font=\normalsize,inner sep=.5mm, thick,every loop/.append style={-latex'}]
    \everymath{\scriptstyle}

    \node[carre,jaune]  (0)  at (0, 0) {};
    \node[rond,jaune]  (1) at (1.6, 0) {};
    \node[rond,jaune]  (2) at (3.2, 0) {};
    
    \coordinate[shift={(0mm,5mm)}] (from) at (0.north);
    \coordinate[shift={(0mm,-5mm)}] (to) at (0.south);
    \path
    (1) edge [loop below, out=240, in=300,looseness=2, distance=8mm,-latex'] node [below] {$\scriptstyle (-1,0)$} (1)
    (2) edge [loop below, out=240, in=300,looseness=2, distance=8mm,-latex'] node [below] {$\scriptstyle (0,0)$} (2);
	\draw[-latex',dotted] (from) to (0);
	\draw[-latex',dotted] (0) to (to);
	\draw[-latex'] (0) to node[above,yshift=1mm] {$\scriptstyle (0,-1)$} (1);
	\draw[-latex'] (1) to node[above,yshift=1mm] {$\scriptstyle (0,0)$} (2);
      \end{tikzpicture}}}
      
	\caption{Gadgets for the reduction from the 2CM halting problem to lower-bounded average-energy games.}
	\label{fig:gadgets}
\end{figure}
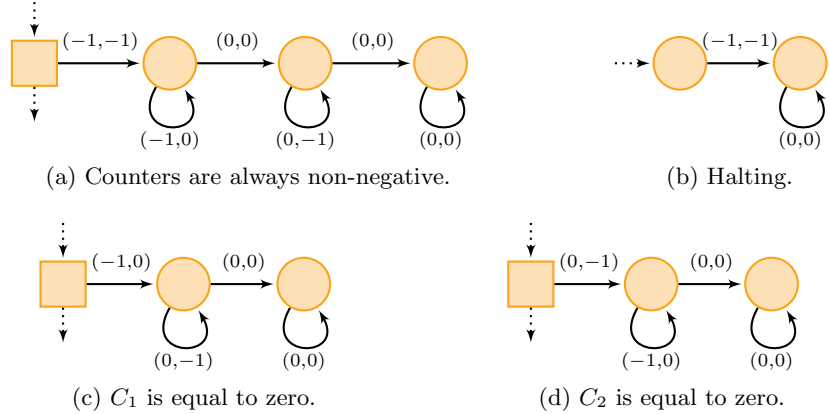
\textit{Increments and decrements} of counters are easily simulated as presented above.

\textit{Values of counters may never go below zero.} To ensure this, we allow $\pI$ to branch after every step of the 2CM simulation to a special gadget represented in Subfig.~\ref{fig:gadgets}\subref{fig:gadgetPos}. If $\pO$ cheats by reaching a negative value on counter $C_1$ or $C_2$, $\pI$ can easily win by branching: indeed, say $C_1$ has value $v_1 < 0$, then the energy level on the first dimension before branching is a most zero and because of the $(-1, -1)$ edge, the energy becomes negative and $\pO$ loses. On the contrary, if $\pO$ does not cheat and maintains both counters non-negative at all times, taking the $(-1, -1)$ edge is safe and the energy level after this edge is $(v_1, v_2)$. Then, $\pO$ can win the game by taking the first loop $v_1$ times, the second one $v_2$ times, and finally going to the absorbing state with energy level $(0,0)$. Indeed, because of the decomposition lemma of~\cite{DBLP:journals/iandc/Chatterjee0RR15}, the average-energy of the resulting play will be equal to $(0,0)$.

\textit{Zero-tests are correctly executed.} Consider counter $C_1$. To ensure that $\pO$ does not cheat by claiming that $C_1$ has value zero while it is strictly positive, we give the possibility to $\pI$ to check zero-tests on $C_1$ by branching using the gadget in Subfig.~\ref{fig:gadgets}\subref{fig:gadgetZeroOne}. Assume $C_1$ has value $v_1 > 0$ and $\pO$ claims it has value zero. Then, $\pI$ branches and after the $(-1,0)$ edge, the energy level on the first dimension is still strictly positive. Since $\pO$ can never decrease it afterwards, he will not meet the threshold $(0,0)$ on the average-energy and will lose the play. Now, assume $C_1$ has indeed value $v_1 = 0$. If $\pI$ decides to branch nonetheless, the energy level after branching is $(0, v_2)$ and $\pO$ can win by looping $v_2$ times in the first cycle and then going to the absorbing state. For counter $C_2$, we use the symmetric gadget depicted in Subfig.~\ref{fig:gadgets}\subref{fig:gadgetZeroTwo}. Similarly, we can check that $\pO$ does not cheat by claiming that $C_1$ (resp.~$C_2$) is strictly positive while it is not. To do so, we give the possibility to $\pI$ to branch after such a claim and decrement $C_1$ (resp.~$C_2$) using an edge $(-1, 0)$ (resp.~$(0,-1)$) and then go to the gadget in Subfig.~\ref{fig:gadgets}\subref{fig:gadgetPos}: if $C_1$ (resp.~$C_2$) had value zero, then the play is losing because the energy level drops below zero, and if $C_1$ (resp.~$C_2$) was strictly positive, $\pO$ can win as described before. 

Therefore, if $\pO$ does not faithfully simulate $\calM$, he is guaranteed to lose in~$\Game$. On the other hand, if $\pI$ stops a faithful simulation, $\pO$ is guaranteed to win.

It remains to argue that $\pO$ wins if, and only if, the machine halts. Indeed, if the machine $\calM$ halts, then $\pO$ simulates its execution faithfully and either he is interrupted and wins, or the simulation ends in the gadget depicted in Subfig.~\ref{fig:gadgets}\subref{fig:gadgetHalt} and he also wins. Indeed, given that this gadget can only be reached with values of counters equal to zero (by hypothesis on $\calM$ and by our assumption on $\calM$), the energy level reaches $(0,0)$ after the $(-1,-1)$ edge, and using the decomposition lemma mentioned above, the average-energy for the play is $(0,0)$. Hence, $\pO$ wins. On the opposite, if $\calM$ does not halt, $\pO$ has no way to reach the halting gadget by means of a faithful simulation, hence if $\pO$ never cheats and $\pI$ never branches, the energy level in both dimensions is at all times at least equal to one (since it is equal to $(v_1 + 1, v_2 + 1)$): the average-energy will be at least $(1, 1)$, hence above the threshold, and $\pO$ loses.

Consequently, we have that $\pO$ wins if, and only if, the 2CM halts, which implies undecidability for lower-bounded average-energy games with two or more dimensions.\qed
\end{proof}

\section{Conclusion}
\label{sec_conc}
We have presented the first algorithm for solving average-energy games with only a lower bound, but no upper bound, on the energy level, thereby solving an open problem from~\cite{BMRLL16}. The algorithm is based on the first upper bound on the necessary memory to implement a winning strategy for $\pO$ in such a game, which solves another open problem from~\cite{BMRLL16}. 

Our algorithm has a doubly-exponential running time, which we complemented by showing the problem to be $\EXPSPACE{}$-hard. This is an improvement over the previous $\EXPTIME{}$-hardness result~\cite{BMRLL16}. An obvious open problem concerns closing this gap. On the other hand, a game due to Hunter~\cite{Hunter14arxiv} shows our doubly-exponential upper on the memory to be almost tight. Finally, for a unary encoding of the weights and the threshold (or if they are polynomially bounded in the size of the graph), our algorithm runs in exponential time and the upper bound on the memory is exponential as well.

Finally, we also investigated multi-dimensional average-energy games in various settings: we showed such games to be undecidable, both for the case without any bounds on the energy level and for the case of only lower bounds. In contrast, the case of games with both a lower and an upper bound in every dimension is in $\ComplexityFont{NEXPTIME}\cap \ComplexityFont{coNEXPTIME}$ and trivially $\EXPTIME{}$-hard, another gap to be closed in further research. Also, we left open the decidability of two-dimensional average-energy games without bounds. 

Our results where obtained by solving a certain type of mean-payoff (one-counter) pushdown game with unbounded weight function. In further research, we investigate whether our techniques are able to solve arbitrary mean-payoff pushdown games with such a weight function.

\bibliographystyle{plain}
\bibliography{biblio}

\end{document}